\documentclass[journal,12pt,draftclsnofoot,onecolumn]{IEEEtran}
\ifCLASSINFOpdf
\else
\fi

\pdfoutput=1

\usepackage[latin1]{inputenc}
\usepackage{indentfirst}
\usepackage{amsfonts}
\usepackage{graphicx}
\usepackage{amssymb}
\usepackage{rotate}
\usepackage{multicol}
\usepackage{epsf}
\usepackage{epsfig}
\usepackage{amsmath}
\usepackage[bookmarks=false]{hyperref}
\usepackage{amsthm}
\usepackage{psfrag}
\usepackage[mathscr]{eucal}
\usepackage{color}
\usepackage{amsmath, amsthm, amssymb, amsthm}
\usepackage{stfloats}
\usepackage{ifthen}
\usepackage{amsmath}
\usepackage{amssymb}
\usepackage{psfrag}
\usepackage{caption}
\usepackage{subcaption}
\usepackage{cite}
\usepackage{bbm}
\usepackage{exscale,relsize}

\usepackage{makecell, pict2e}

\usepackage{enumerate}
\newtheorem{theorem}{Theorem}
\newtheorem{proposition}{Proposition}
\newtheorem{lemma}{Lemma}

\newtheorem{remark}{Remark}

\newenvironment{proof-sketch}{\noindent \textit{Sketch of the Proof:}}{\hfill$\Box$}

\DeclareMathAlphabet{\eurm}{U}{eur}{m}{n}

\newcommand{\E}[2][]{{\mathbb{E}_{#1}} \left[#2\right]}

\begin{document}

\title{Intermittent Communication}
\author{Mostafa~Khoshnevisan,~\IEEEmembership{Student Member,~IEEE,}
        and~J.~Nicholas~Laneman,~\IEEEmembership{Senior Member,~IEEE}
\thanks{This work was supported in part by NSF grants CCF05-46618, CCF11-17365, and CPS12-39222.}
\thanks{This paper was presented in part at ISIT 2012, Allerton 2012, and ITA 2013.}
\thanks{Mostafa Khoshnevisan and J. Nicholas Laneman are with the Department of Electrical Engineering, University of Notre Dame, Notre Dame, IN, 46556 USA, (e-mails: \{mkhoshne, jnl\}@nd.edu).} 
}

%
\maketitle

\begin{abstract}
We formulate a model for intermittent communication that can capture bursty transmissions or a sporadically available channel, where in either case the receiver does not know a priori when the transmissions will occur. Focusing on the point-to-point case, we develop a decoding structure, decoding from pattern detection, and its achievable rate for such communication scenarios. Decoding from pattern detection first detects the locations of codeword symbols and then uses them to decode. We introduce the concept of partial divergence and study some of its properties in order to obtain stronger achievability results. As the system becomes more intermittent, the achievable rates decrease due to the additional uncertainty about the positions of the codeword symbols at the decoder. Additionally, we provide upper bounds on the capacity of binary noiseless intermittent communication with the help of a genie-aided encoder and decoder. The upper bounds imply a tradeoff between the capacity and the intermittency rate of the communication system, even if the receive window scales linearly with the codeword length.
\end{abstract}


\section{Introduction}

\IEEEPARstart{C}{ommunication} systems are traditionally analyzed assuming contiguous transmission of encoded symbols through the channel. However, in many practical applications such an assumption may not be appropriate, and transmitting a codeword can be intermittent due to lack of synchronization, shortage of transmission energy, or burstiness of the system. The challenge is that the receiver  may not explicitly know whether a given output symbol of the channel is the result of sending a symbol of the codeword or is simply a noise symbol containing no information about the message. This paper provides one model, called \textit{intermittent communication} for non-contiguous transmission of codewords in such settings.

If the intermittent process is considered to be part of the channel behavior, then intermittent communication models a sporadically available channel in which at some times a symbol from the codeword is sent, and at other times the receiver observes only noise. The model can be interpreted as an insertion channel in which some number of noise symbols are inserted between the codeword symbols. As another application, if the intermittent process is considered to be part of the transmitter, then we say that the transmitter is intermittent. Practical examples include energy harvesting systems in which the transmitter harvests energy usually from a natural source and uses it for transmission. Assuming that there is a special input that can be transmitted with zero energy, the transmitter sends the symbols of the codeword if there is enough energy for transmission, and sends the zero-cost symbol otherwise.

\subsection{Related Work}
\label{related}

Conceptually, two sets of related work are bit-synchronous-frame-asynchronous and bit-asynchronous communication. The former corresponds to contiguous transmission of codeword symbols in which the receiver observes noise before and after transmission, and the latter corresponds to a channel model in which some number of  symbols are inserted at the output of the channel, or some of the input symbols of the channel are deleted. A key assumption of these communication models is that the receiver does not know a priori the positions of codeword / noise / deleted symbols, capturing certain kinds of asynchronism. Generally, this asynchronism makes the task of the decoder more difficult since, in addition to the uncertainty about the message of interest, there is also uncertainty about the positions of the symbols.

An information theoretic model for bit-synchronous-frame-asynchronous communication is developed in~\cite{asynch-wornell-1, asynch-wornell, asynch-per-unit-tse, asynch-polyanskiy}  with a single block transmission of length $k$ that starts at a random time $v$ unknown to the receiver, within an exponentially large window $n=e^{\alpha k}$ known to the receiver, where $\alpha$ is called the asynchronism exponent. In this model, the transmission is contiguous; once it begins, the whole codeword is transmitted, and the receiver observes noise only both before and after transmission. Originally, the model appears in~\cite{frame-wornell} with the goal being to locate a sync pattern at the output of a memoryless channel using a sequential detector without communicating a message to the receiver. The same authors extend the framework for modeling asynchronous communication in~\cite{asynch-wornell-1, asynch-wornell}, where communication rate is defined with respect to the decoding delay, and give bounds on the capacity of the system.

Capacity per unit cost for asynchronous communication is studied in~\cite{asynch-per-unit-tse}. The capacity of asynchronous communication with a simpler notion of communication rate, i.e., with respect to the codeword length, is a special case of the result in~\cite{asynch-per-unit-tse}. In~\cite{asynch-polyanskiy}, it is shown that if the decoder is required to both decode the message and locate the codeword exactly, the capacity remains unchanged, and can be universally achieved. Also, it is shown that the channel dispersion is not affected because of asynchronism. In a recent work~\cite{energy-sample-asynch}, the authors study the capacity per unit cost for asynchronous communication if the receiver is constrained to sample only a fraction of the channel outputs, and they show that the capacity remains unchanged under this constraint.

In~\cite{isit-asynch}, a slotted asynchronous channel model is investigated and fundamental limits of asynchronous communication in terms of missed detection and false alarm error exponents are studied. A more general result in the context of single-message unequal error protection (UEP) is given in~\cite{isit-unequal-error}.

Although we formally introduce the system model of intermittent communication in Section~\ref{sec-system-model}, we would like to briefly compare asynchronous communication~\cite{frame-wornell} with intermittent communication. First, as contrasted in Figure~\ref{fig-asynch-inter-compare}, the transmission of codeword symbols is contiguous in asynchronous communication, whereas it is bursty in intermittent communication. Second, the receive window $n$ is exponential with respect to the codeword length $k$ in asynchronous communication, but as we will see it is linear in intermittent communication. Finally, in both models, we assume that the receiver does not know a priori the positions of codeword symbols.

\begin{figure}
        \centering
        \begin{subfigure}[t]{.8\textwidth}
                \centering
                \includegraphics[width=\textwidth]{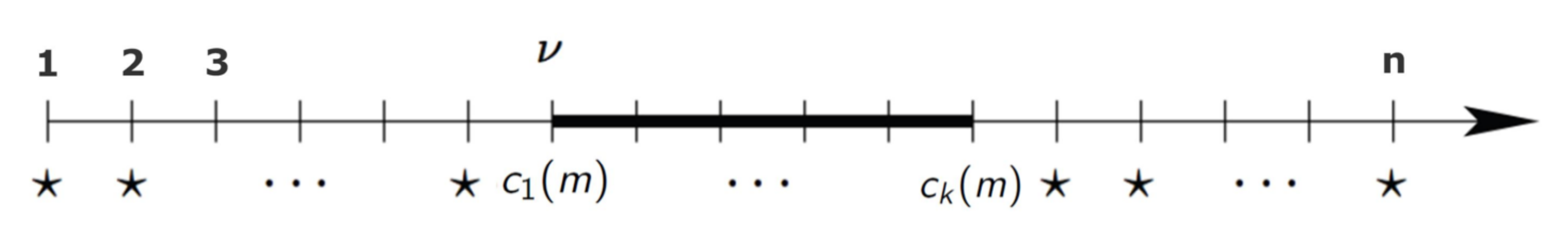}
                \caption{Asynchronous communication~\cite{frame-wornell}: $n=e^{\alpha k}$.}
                \label{fig-asynch}
        \end{subfigure}

        \begin{subfigure}[t]{.8\textwidth}
                \centering
                \includegraphics[width=\textwidth]{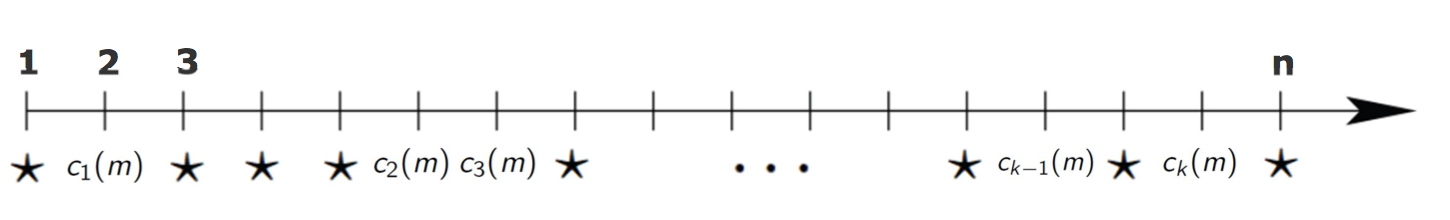}
                \caption{Intermittent communication: $n=\alpha k$.}
                \label{fig-intermittent}
        \end{subfigure}
        \caption{Comparison of the channel input sequence of asynchronous communication~\cite{frame-wornell} with intermittent communication.}\label{fig-asynch-inter-compare}
\end{figure}

Non-contiguous transmission of codeword symbols in bit-asynchronous communication can be described by the following insertion channel: after the $i^{th}$ symbol of the codeword, $N_i$ fixed noise symbols $\star$ are inserted, where $N_i$, $i=1,2,...,k$ are random variables, possibly independent and identically distributed (iid). The resulting sequence passes through a discrete memoryless channel, and the receiver should decode the message based on the output of the channel without knowing the positions of the codeword symbols. As will see later, if $N \ge k$ is the random variable denoting the total number of received symbols, the intermittency rate $N/k \xrightarrow{\text{ } p \text{ }} \alpha$ as $k \to \infty$ will be an important parameter of the system. To the best of our knowledge, this insertion channel model has not been studied before.

A more general class of channels with synchronization errors is studied in~\cite{dobrushin-general} in which every transmitted symbol is independently replaced with a random number of the same symbol, possibly empty string to model a deletion event. Dobrushin~\cite{dobrushin-general} proved the following characterization of the capacity for such iid synchronization error channels.

\begin{theorem}
\label{dobrushin}
(\cite{dobrushin-general}): For iid synchronization error channels, let $X^k:=(X_1,X_2,...,X_k)$ denote the channel input sequence of length $k$, and $Y^N:=(Y_1,Y_2,...,Y_N)$ denote the corresponding output sequence at the decoder, where length $N$ is a random variable determined by the channel realization. The channel capacity is
\begin{equation} \label{e: dobrushin}
C=\lim_{k \to \infty} \max_{P_{X^k}} \frac{1}{k} \mathbb{I}(X^k;Y^N),
\end{equation}
\end{theorem}
where $\mathbb{I}(X;Y)$ denotes the average mutual information~\cite{Cover}. There are two difficulties related to computing the capacity through this characterization. First, it is challenging to compute the mutual information because of the memory inherent in the joint distribution of the input and output sequences. Second, the optimization over all the input distributions is computationally involved. A single-letter characterization for the capacity of the general class of synchronization error channels is still an open problem, even though there are many papers deriving bounds on the capacity of the insertion / deletion channels~\cite{deletion-insertion-survey, gallager-synch-channel, diggavi-deletion-achv, improved-lower-drinea, direct-lower-drinea, diggavi-outer-deletion, duman-outer-deletion1, duman-outer-deletion2, sticky, insertion}.

Focusing on the lower bounds, Gallager~\cite{gallager-synch-channel} considers a channel model with substitution and insertion / deletion errors and derives a lower bound on the channel capacity. In~\cite{diggavi-deletion-achv}, codebooks from first-order Markov chains are used to improve the achievability results for deletion channels. The intuition is that it is helpful to put some memory in the codewords if the channel has some inherent memory. Improved achievability results are obtained in~\cite{improved-lower-drinea, direct-lower-drinea}. in~\cite{direct-lower-drinea}, authors provide a direct lower bound on the information capacity given by~\eqref{e: dobrushin} for channels with iid deletions and duplications with an input distribution following a symmetric, first-order Markov chain.

Upper bounds  on the capacity of iid insertion / deletion channels are obtained in~\cite{duman-outer-deletion1, duman-outer-deletion2, diggavi-outer-deletion}, in which the capacity of an auxiliary channel obtained by a genie-aided decoder (and encoder) with access to side-information about the insertion / deletion process is numerically evaluated. Although our insertion channel model is different, we can apply some of these ideas and techniques to upper bound the capacity of intermittent communication.

\subsection{Summary of Contributions}
\label{summary of cont}

After introducing a model for intermittent communication in Section~\ref{sec-system-model}, we develop a coding theorem for achievable rates to lower bound the capacity in Section~\ref{sec-achv}. Toward this end, we introduce the notion of partial divergence and its properties. By using decoding from pattern detection, we obtain achievable rates that are also valid for arbitrary intermittent processes.

Focusing on the binary-input binary-output noiseless case, we obtain upper bounds on the capacity of intermittent communication in Section~\ref{sec-upper} by providing the encoder and the decoder with various amounts of side-information, and calculating or upper bounding the capacity of this genie-aided system. Although the gap between the achievable rates and upper bounds is fairly large, especially for large values of intermittency rate, the results suggest that linear scaling of the receive window with respect to the codeword length considered in the system model is relevant since the upper bounds imply a tradeoff between the capacity and the intermittency rate.

\section{System Model and Foundations}
\label{sec-system-model}

We consider a communication scenario in which a transmitter communicates a single message $m \in \{1,2,...,e^{kR}=M\}$ to a receiver over a discrete memoryless channel (DMC) with probability transition matrix $W$ and input and output alphabets $\mathcal{X}$ and $\mathcal{Y}$, respectively. Let $C_W$ denote the capacity of the DMC. Also, let $\star \in \mathcal{X}$ denote the ``noise symbol'' or the ``idle symbol'', which corresponds to the input of the channel when the transmitter is ``silent''. The transmitter encodes the message as a codeword $c^k(m)$ of length $k$,  which is the input sequence of intermittent process shown in Figure~\ref{fig-system-model}.

\begin{figure}[t]
\centerline{\scalebox{.53}{\includegraphics{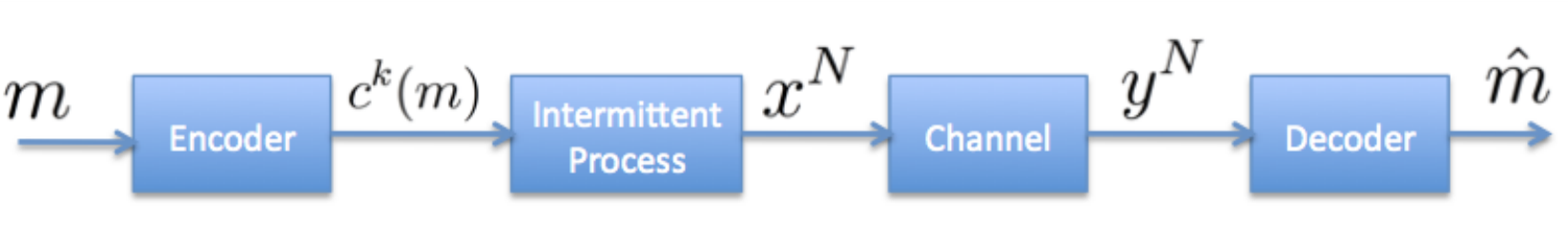}}}
\caption{System model for intermittent communication.}
\label{fig-system-model}
\end{figure}

The intermittent process captures the burstiness of the channel or the transmitter and can be described as follows: After the $i^{th}$ symbol from the codeword, $N_i$ noise symbols $\star$ are inserted, where $N_i$'s are iid geometric random variables with mean $\alpha-1$, where $\alpha \ge 1$ is  the \textit{intermittency rate}. Note that $i=0,1,...,k$, where $N_0$ denotes the number of noise symbols inserted before the first symbol of the codeword. As we will see later, if $N \ge k$ is the random variable denoting the total number of received symbols, the intermittency rate $N/k \xrightarrow{\text{ } p \text{ }} \alpha$ as $k \to \infty$ will be an important parameter of the system. In fact, the larger the value of $\alpha$, the larger the receive window, and therefore, the more intermittent the system becomes with more uncertainty about the positions of the codeword symbols; if $\alpha=1$, the system is not intermittent and corresponds to contiguous communication. We call this scenario \textit{intermittent communication} and denote it by the tuple $(\mathcal{X},\mathcal{Y},W, \star,\alpha)$.

This model corresponds to an iid insertion channel model in which at each time slot a codeword symbol is sent with probability $p_t:=1/\alpha$ and the noise symbol $\star$ is sent with probability $1-p_t$. The output of the intermittent process then goes through a DMC. At the decoder, there are $N$ symbols, where $N$ is a random variable with distribution:
\begin{equation} \label{e-neg-binom}
P(N=n)=\binom{n}{k} p_t^{k+1} (1-p_t)^{n-k}, n \ge k,
\end{equation}
with $\E N=\alpha k +\alpha -1$, and we have
\begin{equation} \label{e-conv-prob}
\frac{N}{k}=\frac{k+N_0+N_1+N_2+...+N_k}{k} \xrightarrow{\text{ } p \text{ }} 1+\mathbb{E}(N_0)=\alpha, \text{ as } k \to \infty.
\end{equation}

Therefore, the receive window $N$ scales linearly with the codeword length $k$, as opposed to the exponential scaling in asynchronous communication~\cite{asynch-wornell-1, asynch-wornell, asynch-per-unit-tse, asynch-polyanskiy}. The intermittent communication model represents bursty communication in which either the transmitter or the channel is bursty. In a bursty communication scenario, the receiver usually does not know the realization of the bursts. Therefore, we assume that the receiver does not know the positions of the codeword symbols, making the decoder's task more involved. However, we assume that the receiver knows the codeword length $k$ and the intermittency rate $\alpha$.

Denoting the decoded message by $\hat{m}$, which is a function of the random sequence $Y^N$, and defining a code as in~\cite{Cover}, we say that rate $R$ is achievable if there exists a sequence of length $k$ codes of size $M=e^{kR}$ with average probability of error $\frac{1}{M}\sum_{m=1}^{M}\mathbb{P} (\hat{m} \ne m) \to 0$ as $k \to \infty$. Note that the code rate is defined as $\log M /k$. The capacity is the supremum of all the achievable rates. Rate region $(R,\alpha)$ is said to be achievable if the rate $R$ is achievable for the corresponding scenario with the intermittency rate $\alpha$.

At this point, we summarize several notations that are used throughout the sequel. We use $o(\cdot)$ and $\text{poly}(\cdot)$ to denote quantities that grow strictly slower than their arguments and are polynomial in their arguments, respectively. By $X \sim P(x)$, we mean $X$ is distributed according to $P$. For convenience, we define $W_\star(\cdot):=W(\cdot|X=\star)$, and more generally, $W_x(\cdot):=W(\cdot|X=x)$. In this paper, we use the convention that $\binom{n}{k}=0$ if $k <0$ or $n <k$, and the entropy $H(P)=-\infty$ if $P$ is not a probability mass function, i.e., one of its elements is negative or the sum of its elements is larger than one. We also use the conventional definitions $x^+:=\max\{x,0\}$, and if $0 \le \rho \le 1$, then $\bar{\rho}:=1-\rho$. Additionally, $\doteq$ denotes an equality in exponential sense as $k \to \infty$, i.e., $\lim_{k \to \infty} \frac{1}{k} \log (\cdot)$ of both sides are equal. Unless otherwise stated, we use capital letters for random variables, and small letters for their realizations.

It can be seen that the result of Theorem~\ref{dobrushin} for iid synchronization error channels applies to intermittent communication model, and therefore, the capacity equals $\lim_{k \to \infty} \max_{P_{C^k}} \frac{1}{k} \mathbb{I}(C^k;Y^N)$. Denoting the binary entropy function by $h(p):=-p\log p-(1-p) \log (1-p)$, we have the following theorem.

\begin{theorem}
\label{theorem-simple-lower}
For intermittent communication $(\mathcal{X},\mathcal{Y},W, \star,\alpha)$, rates less than \\ $R_1:=\left(C_W-\alpha h(1/\alpha)\right)^+$ are achievable.
\end{theorem}

\begin{proof}
We show that $R_1$ is a lower bound for the capacity of intermittent communication by lower bounding the mutual information. Let vector $T^{k+1}:=(N_0,N_1,...,N_k)$ denote the number of noise insertions in between the codeword symbols, where the $N_i$'s are iid geometric random variables with mean $\alpha-1$: $$P(N_i=n_i)=(1-\frac{1}{\alpha})^{n_i}\frac{1}{\alpha}, n_i \ge 0.$$
Now, we have
\begin{align}
\mathbb{I}(C^k;Y^N)&= \mathbb{I}(C^k;Y^N,T^{k+1})- \mathbb{I}(C^k;T^{k+1}|Y^N) \notag \\
& =  \mathbb{I}(X^k;Y^k)- \mathbb{I}(C^k;T^{k+1}|Y^N) \label{eq-thm-simple-1} \\
& \ge k \mathbb{I}(X;Y) -H(T^{k+1}) \label{eq-thm-simple-2} \\
&= k \mathbb{I}(X;Y) -(k+1)H(N_0)  \label{eq-thm-simple-3} \\
&= k \mathbb{I}(X;Y) -(k+1)\alpha h(\frac{1}{\alpha}),  \label{eq-thm-simple-4}
\end{align}
where~\eqref{eq-thm-simple-1} follows from $\mathbb{I}(C^k;Y^N,T^{k+1})=\mathbb{I}(C^k;T^{k+1})+\mathbb{I}(C^k;Y^N|T^{k+1})=\mathbb{I}(C^k;Y^N|T^{k+1})=\mathbb{I}(X^k;Y^k)$, where the first equality follows from the chain rule, the second equality follows from the fact that the codeword $C^k$ is independent of the insertion process $T^{k+1}$, and the last equality follows from the fact that conditioned on the positions of noise symbols $T^{k+1}$, the mutual information between $C^k$ and $Y^N$ equals the mutual information between input and output sequences of the DMC without considering the noise insertions; where~\eqref{eq-thm-simple-2} follows  by considering iid symbols of codewords and by the fact that conditioning cannot increase the entropy; and~\eqref{eq-thm-simple-3} and~\eqref{eq-thm-simple-4} follow from the fact that $N_i$'s are iid geometric random variables. Finally, the result follows after dividing both sides by $k$ and considering the capacity achieving input distribution of the DMC.
\end{proof}

Although the lower bound on the capacity of intermittent communication in Theorem~\ref{theorem-simple-lower} is valid for the specific intermittent process described above, our achievability results in Section~\ref{sec-achv} apply to an arbitrary insertion process as long as $N/k \xrightarrow{\text{ } p \text{ }} \alpha$ as $k \to \infty$.

We will make frequent use of the notations and results from the method of types, a powerful technique in large deviation theory developed by Csisz\'{a}r and K\"{o}rner~\cite{csiszar}. Let $\mathcal{P}^\mathcal{X}$ denote the set of distributions over the finite alphabet $\mathcal{X}$.  The empirical distribution (or type) of a sequence $x^n \in \mathcal{X}^n$ is denoted by $\hat{P}_{x^n}$. A sequence $x^n$ is said to have a type $P$ if $\hat{P}_{x^n}=P$. The set of all sequences that have type $P$ is denoted $T_P^n$, or more simply $T_P$. Joint empirical distributions are denoted similarly. The set of sequences $y^n$ that have a conditional type $W$ given $x^n$ is denoted $T_W(x^n)$. A sequence $x^n \in \mathcal{X}^n$ is called $P$-typical with constant $\mu$, denoted $x^n \in T_{[P]_\mu}$, if
\begin{equation} \notag
| \hat{P}_{x^n}(x) - P(x)| \le \mu \text{ for every } x \in \mathcal{X}.
\end{equation}
Similarly, a sequence $y^n \in \mathcal{Y}^n$ is called $W$-typical conditioned on $x^n \in \mathcal{X}^n$ with constant $\mu$, denoted $y^n \in T_{[W]_\mu}(x^n)$, if
\begin{equation} \notag
| \hat{P}_{x^n,y^n}(x,y) - \hat{P}_{x^n}(x) W(y|x) | \le \mu \text{ for every } (x,y) \in  \mathcal{X} \times  \mathcal{Y}.
\end{equation}

For $P, P^\prime \in \mathcal{P}^{\mathcal{X}}$ and $W, W^\prime \in \mathcal{P}^{\mathcal{Y}|\mathcal{X}}$, the Kullback-Leibler divergence between $P$ and $P^\prime$ is defined as
\begin{equation} \notag
D(P\|P^\prime):=\sum_{x \in \mathcal{X}} P(x) \log \frac{P(x)}{P^\prime(x)},
\end{equation}
and the conditional information divergence between $W$ and $W^\prime$ conditioned on $P$ is defined as
\begin{equation} \notag
D(W\|W^\prime|P):=\sum_{x \in \mathcal{X}} P(x) \sum_{y \in \mathcal{Y}} W(y|x) \log \frac{W(y|x)}{W^\prime(y|x)}.
\end{equation}
The average mutual information between $X \sim P$ and $Y \sim PW$ and coupled via $P_{Y|X}=W$ is denoted by $\mathbb{I}(P,W)$. With these definitions, we now state the following lemmas, which are used throughout the paper.

\begin{lemma} \label{fact-typ1}
(\cite[Lemma 1.2.6]{csiszar}): If $X^n$ is an iid sequence according to $P^\prime$, then the probability that it has a type $P$ is bounded by
\begin{equation} \notag
\frac{1}{(n+1)^{|\mathcal{X}|}}e^{-nD(P\|P^\prime)} \le \mathbb{P}(X^n \in T_P) \le e^{-nD(P\|P^\prime)}.
\end{equation}
Also, if the input $x^n \in \mathcal{X}^n$ to a memoryless channel $W^\prime \in \mathcal{P}^{\mathcal{Y}|\mathcal{X}}$ has type $P$, then the probability that the observed channel output sequence $Y^n$ has a conditional type $W$ given $x^n$ is bounded by
\begin{equation} \notag
\frac{1}{(n+1)^{|\mathcal{X}\|\mathcal{Y}|}}e^{-nD(W\|W^\prime|P)} \le \mathbb{P}(Y^n \in T_W(x^n)) \le e^{-nD(W\|W^\prime|P)}.
\end{equation}
\end{lemma}

\begin{lemma} \label{fact-typ2}
(\cite[Lemma 1.2.12]{csiszar}): If $X^n$ is an iid sequence according to $P$, then
\begin{equation} \notag
\mathbb{P}(X^n \in T_{[P]_\mu}) \ge 1- \frac{|\mathcal{X}|}{4n\mu^2}.
\end{equation}
Also, if the input $x^n \in \mathcal{X}^n$ to a memoryless channel $W \in \mathcal{P}^{\mathcal{Y}|\mathcal{X}}$, and $Y^n$ is the output, then
\begin{equation} \notag
\mathbb{P}(Y^n \in T_{[W]_\mu}(x^n)) \ge 1- \frac{|\mathcal{X}\|\mathcal{Y}|}{4n\mu^2},
\end{equation}
and here the terms subtracted from $1$ could be replaced by exponentially small terms $2|\mathcal{X}| e^{-2n\mu^2}$ and $2|\mathcal{X}\|\mathcal{Y}| e^{-2n\mu^2}$, respectively.
\end{lemma}

Finally, we state a stronger version of the packing lemma~\cite[Lemma 3.1]{ElGamal-Kim} that will be useful in typicality decoding, and is proved in~\cite[Equations (24) and (25)]{asynch-wornell} based on the method of types.

\begin{lemma} \label{fact1}
Assume that $X^n$ is an iid sequence according to $P$, and $\tilde{Y}^n$ is an arbitrary sequence, where $\tilde{Y} \in \mathcal{P}^\mathcal{Y}$. If $\tilde{Y}^n$ and $X^n$ are independent, then
\begin{equation} \notag
\mathbb{P}(\tilde{Y}^n \in T_{[W]_\mu}(X^n)) \le \text{poly}(n)e^{-n(\mathbb{I}(P,W)-\epsilon)}
\end{equation}
for all $n$ sufficiently large, where $\epsilon > 0$ can be made arbitrarily small by choosing a small enough typicality parameter $\mu$.
\end{lemma}

It is important to note that the random sequence $\tilde{Y}^n$ in Lemma~\ref{fact1} does not have to be iid. In addition, $\tilde{Y}$ does not have to have the same distribution as $PW$.

\section{Achievability}
\label{sec-achv}

In this section, we obtain achievability results for intermittent communication based upon \textit{decoding from pattern detection}, which attempts to decode the transmitted codeword only if the selected outputs appear to be a pattern of codeword symbols. In order to analyze the probability of error this decoding structure, which gives a larger and more general achievable rate than the one in Theorem~\ref{theorem-simple-lower}, we use a generalization of Sanov's Theorem that leads to the notion of \textit{partial divergence} and its properties described in Section~\ref{subsec-partial-divergence}. Although the system model in Section~\ref{sec-system-model} assumes iid geometric insertions, the results of this section apply to a general intermittent process, as we point out in Remark~\ref{remark-achv-general}.

\subsection{Partial Divergence}
\label{subsec-partial-divergence}

We will see in Section~\ref{subsec-decode-algorithm} that a relevant function is $d_\rho(P\|Q)$, which we call partial divergence. In this section, we examine some of the interesting properties of partial divergence that provide insights about the achievable rates in Section~\ref{subsec-decode-algorithm}. Qualitatively speaking, partial divergence is the normalized exponent of the probability that a sequence with independent elements generated partially according to one distribution and partially according to another distribution has a specific type. This exponent is useful in characterizing a decoder's ability to distinguish a sequence obtained partially from the codewords and partially from the noise from a codeword sequence or a noise sequence. The following lemma is a generalization of Sanov's Theorem~\cite{csiszar}.

\begin{lemma} \label{lemma1}
Consider the  distributions $P,Q,Q^\prime \in \mathcal{P}^\mathcal{X}$ on a finite alphabet $\mathcal{X}$. A random sequence $X^k$ is generated as follows: $\rho k$ symbols iid according to $Q$ and $\bar{\rho} k$ symbols iid according to $Q^\prime$, where $0 \le \rho \le 1$. The normalized exponent of the probability that $X^k$ has type $P$ is
\begin{align}
d(P,Q,Q^\prime,\rho)&:=\lim_{k \to \infty} -\frac{1}{k} \log \mathbb{P}(X^k \in T_P)   \notag \\
&=\min_{\substack{P_1,P_2 \in \mathcal{P}^\mathcal{X}: \\ \rho P_1+\bar{\rho} P_2=P}} \rho D(P_1\|Q) +\bar{\rho} D(P_2\|Q^\prime).  \label{eq-partial-lemma}
\end{align}
\end{lemma}

\begin{proof}
With some abuse of notation, let $X_1^{\rho k}$ and $X_2^{\bar{\rho} k}$ be the sequence of symbols in $X^k$ that are iid according to $Q$, $Q^\prime$, respectively. If these sequences have types $P_1$ and $P_2$, respectively, then the whole sequence $X^k$ has type $\rho P_1+\bar{\rho} P_2$. Therefore, we have
\begin{align}
\mathbb{P}(X^k \in T_P) &=\mathbb{P}\left(\cup_{\substack{P_1,P_2 \in \mathcal{P}^\mathcal{X}: \\ \rho P_1+\bar{\rho} P_2=P}} \{X_1^{\rho k} \in T_{P_1}, X_2^{\bar{\rho} k} \in T_{P_2}\}\right) \label{e-proof-lemma-gen-0} \\
&=\sum_{\substack{P_1,P_2 \in \mathcal{P}^\mathcal{X}: \\ \rho P_1+\bar{\rho} P_2=P}}\mathbb{P}(X_1^{\rho k} \in T_{P_1}, X_2^{\bar{\rho} k} \in T_{P_2}) \label{e-proof-lemma-gen-1} \\
&\doteq \sum_{\substack{P_1,P_2 \in \mathcal{P}^\mathcal{X}: \\ \rho P_1+\bar{\rho} P_2=P}} \exp \{-k(\rho D(P_1\|Q) +\bar{\rho} D(P_2\|Q^\prime))\}, \label{e-proof-lemma-gen-2}\\
&\doteq \exp \left\{-k \min_{\substack{P_1,P_2 \in \mathcal{P}^\mathcal{X}: \\ \rho P_1+\bar{\rho} P_2=P}} \rho D(P_1\|Q) +\bar{\rho} D(P_2\|Q^\prime) \right\} \label{e-proof-lemma-gen-3}
\end{align}
where:~\eqref{e-proof-lemma-gen-1} follows from the disjointedness of the events in~\eqref{e-proof-lemma-gen-0} since a sequence has a unique type;~\eqref{e-proof-lemma-gen-2} follows from the independence of the events in~\eqref{e-proof-lemma-gen-1} and obtaining the probability of each of them according to Lemma~\ref{fact-typ1} to first order in the exponent; and~\eqref{e-proof-lemma-gen-3} follows from the fact that the number of different types is polynomial in the length of the sequence~\cite{csiszar}, which makes the total number of terms in the summation~\eqref{e-proof-lemma-gen-2} polynomial in $k$, and therefore, the exponent equals the largest exponent of the terms in the summation~\eqref{e-proof-lemma-gen-2}.
\end{proof}

Specializing Lemma~\ref{lemma1} for $Q^\prime=Q$ results in Lemma~\ref{fact-typ1}, and we have $d(P,Q,Q,\rho)=D(P\|Q)$. However, we will be interested in the special case of Lemma~\ref{lemma1} for which $Q^\prime=P$. In other words, we need to upper bound the probability that a sequence has a type $P$ if its elements are generated independently with a fraction $\rho$ according to $Q$ and the remaining fraction $\bar{\rho}$ according to $P$. For this case, we call $$d_\rho(P\|Q):=d(P,Q,P,\rho)$$ the partial divergence between $P$ and $Q$ with mismatch ratio $0 \le \rho \le 1$. Proposition~\ref{prop1} gives an explicit expression for the partial divergence by solving the optimization problem in~\eqref{eq-partial-lemma} for the special case of $Q^\prime=P$.

\begin{proposition} \label{prop1}
If $\mathcal{X}=\{1,2,...,|\mathcal{X}|\}$ and $P,Q \in \mathcal{P}^\mathcal{X}$, where $P:=(p_1,p_2,...,p_{|\mathcal{X}|})$ and $Q:=(q_1,q_2,...,q_{|\mathcal{X}|})$ and we assume that all values of the PMF $Q$ are nonzero, then the partial divergence can be written as
\begin{equation} \label{E: explicit d}
d_\rho(P\|Q)= D(P\|Q)-\sum_{j=1}^{|\mathcal{X}|} p_j \log(c^*+\frac{p_j}{q_j})+\rho \log c^* +h(\rho),
\end{equation}
where $c^*$ is a function of $\rho$, $P$, and $Q$, and can be uniquely determined from
\begin{equation} \label{E: const c}
c^* \sum_{j=1}^{|\mathcal{X}|} \frac{p_j q_j}{c^* q_j+p_j}=\rho.
\end{equation}
\end{proposition}

\begin{proof}
See Appendix~\ref{app-proof-prop1}.
\end{proof}

The next proposition states some of the properties of the partial divergence, which will be applied in the sequel.

\begin{proposition} \label{prop2}
The partial divergence $d_\rho(P\|Q), 0 \le \rho \le 1$,  where all of the elements of the PMF $Q$ are nonzero, has the following properties:
\begin{enumerate}[(a)]
\item $d_0(P\|Q)=0$. \label{a}
\item $d_1(P\|Q)=D(P\|Q)$. \label{b}
\item Partial divergence is zero if $P=Q$, i.e., $d_\rho(P\|P)=0$.
\item Let $d_\rho^\prime(P\|Q):=\frac{\partial d_\rho(P\|Q)}{\partial \rho}$ denote the derivative of the partial divergence with respect to $\rho$, then $d_0^\prime(P\|Q)=0$. \label{c}
\item If $P \ne Q$, then $d_\rho^\prime(P\|Q) > 0$, for all $0 < \rho \le 1$, i.e., partial divergence is increasing in $\rho$. \label{d}
\item If $P \ne Q$, then $d_\rho^{\prime\prime}(P\|Q) > 0$, for all $0 \le \rho \le 1$, i.e., partial divergence is convex in $\rho$. \label{e}
\item $0 \le d_\rho(P\|Q) \le \rho D(P\|Q)$. \label{f}
\end{enumerate}
\end{proposition}

\begin{proof}
See Appendix~\ref{app-proof-prop2}.
\end{proof}

Figure~\ref{F: partial diver} shows two examples of the partial divergence for PMF's with alphabets of size 4. Specifically, $d_\rho(P\|Q)$ versus $\rho$ is plotted for $P=(0.25,0.25,0.25,0.25)$, and two different $Q$'s, $Q_1=(0.1,0.1,0.1,0.7)$ and $Q_2=(0.1,0.4,0.1,0.4)$. The properties in Proposition~\ref{prop2} are apparent in the figure for these examples.

\begin{figure}[t]
\centerline{\scalebox{.33}{\includegraphics{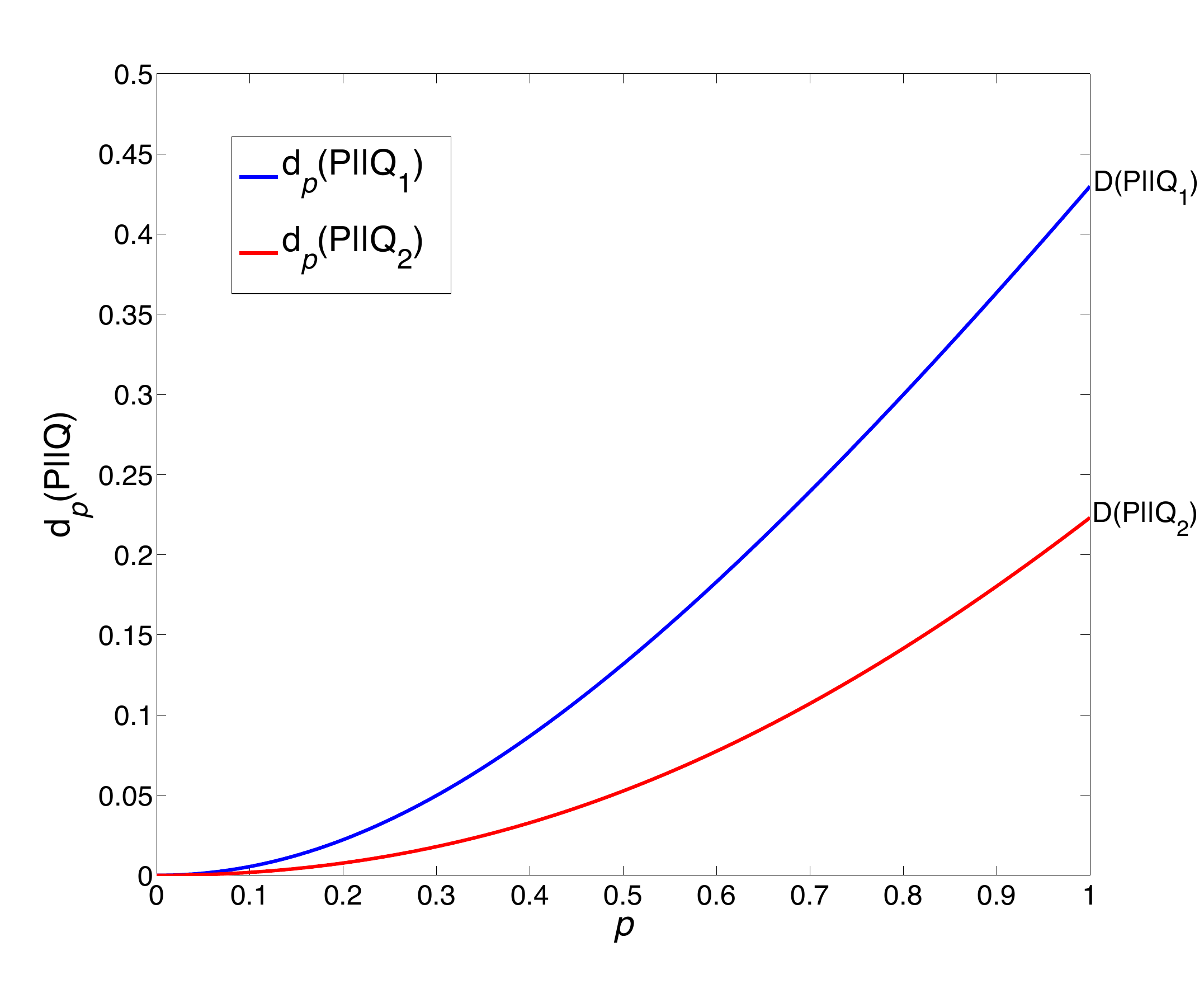}}}
\caption{Partial divergence $d_\rho(P\|Q)$ versus $\rho$ for $P=(0.25,0.25,0.25,0.25)$, $Q_1=(0.1,0.1,0.1,0.7)$, and $Q_2=(0.1,0.4,0.1,0.4)$.}
\label{F: partial diver}
\end{figure}

\begin{proposition} \label{prop2-2}
The partial divergence $d_\rho(P\|Q), 0 \le \rho \le 1$,  satisfies
$$d_\rho(P\|Q) \ge D(P\| \rho Q+\bar{\rho} P).$$
\end{proposition}

\begin{proof}
From the definition of the partial divergence and~\eqref{eq-partial-lemma}, we have
\begin{align}
d_\rho(P\|Q)&=\min_{\substack{P_1,P_2 \in \mathcal{P}^\mathcal{X}: \\ \rho P_1+\bar{\rho} P_2=P}} \rho D(P_1\|Q) +\bar{\rho} D(P_2\|P) \notag \\
& \ge \min_{\substack{P_1,P_2 \in \mathcal{P}^\mathcal{X}: \\ \rho P_1+\bar{\rho} P_2=P}} D(\rho P_1+\bar{\rho} P_2 \| \rho Q+\bar{\rho} P) \label{eq-prop-partial-2-1} \\
& = D(P \| \rho Q+\bar{\rho} P), \label{eq-prop-partial-2-2}
\end{align}
where~\eqref{eq-prop-partial-2-1} follows from the convexity of the Kullback-Leibler divergence~\cite{Cover}; and~\eqref{eq-prop-partial-2-2} follows from the constraint $\rho P_1+\bar{\rho} P_2=P$ in the minimization.
\end{proof}

The interpretation of Proposition~\ref{prop2-2} is that if all the elements of a sequence are generated independently according to a mixture probability $\rho Q+\bar{\rho} P$, then it is more probable that this sequence has type $P$ than in the case that a fraction $\rho$ of its elements are generated independently according to $Q$ and the remaining fraction $\bar{\rho}$ are generated independently according to $P$. Since the partial divergence $d_\rho(P\|Q)$ is used to obtain achievability results in Theorem~\ref{th-achv2-basic}, it can be substituted with the less complicated function $D(P \| \rho Q+\bar{\rho} P)$ in~\eqref{E: basic-f} with the expense of loosening the bound according to Proposition~\ref{prop2-2}.

\subsection{Decoding From Pattern Detection}
\label{subsec-decode-algorithm}

In this section, we introduce decoding from pattern detection. The encoding structure is as follows: Given an input distribution $P$, the codebook $C^k(m), m \in \{1,2,..., M\}$ is randomly and independently generated, i.e., all $C_i(m), i \in \{1,2,...,k\}, m \in \{1,2,..., M\}$ are iid according to $P$. Note that the number of received symbols at the decoder, $N$, is a random variable. However, without loss of generality, we can focus on the case that $|N/k-\alpha|<\epsilon$ for arbitrarily small $\epsilon$, and essentially assume that the receive window is of length $n=\alpha k$, which makes the analysis of the probability of error for the decoding algorithm more concise. In order to see this, we assume that for arbitrarily small $\epsilon$, if $|N/k-\alpha| \ge \epsilon$, then the decoder declares an error. The probability of this error, $\mathbb{P}(|N/k-\alpha| \ge \epsilon)$, approaches zero as $k \to \infty$.

The decoder chooses $k$ of the $n$ symbols from the output sequence $y^n$. Let $\tilde{y}^k$ denote the subsequence of the chosen symbols, and $\hat{y}^{n-k}$ denote the subsequence of the other symbols. Decoding from pattern detection involves two stages for each choice of the output symbols. For each choice, the first stage checks if this choice of the output symbols is a good one, which consists of checking if $\tilde{y}^k$ is induced by a codeword, i.e., if $\tilde{y}^k \in T_{PW}$, and if $\hat{y}^{n-k}$ is induced by noise, i.e., if $\hat{y}^{n-k} \in T_{W_\star}$. If both of these conditions are satisfied, then we proceed to the second stage, which is performing typicality decoding with $\tilde{y}^k$ over the codebook, i.e., checks if $\tilde{y}^k \in T_{[W]_\mu}(c^k(m))$ for a unique index $m$ and a fixed typicality parameter $\mu > 0$. If this condition is not satisfied, then we make another choice for the $k$ symbols and repeat the two-stage decoding procedure. At any step that we run the second stage, if the typicality decoding declares a message as being sent, then decoding ends. If the decoder does not declare any message as being sent by the end of all $\binom{n}{k}$ choices, then the decoder declares an error. In this structure, we constrain the search domain for the typicality decoding (the second stage) only to typical patterns by checking that our choice of codeword symbols satisfies the conditions in the first stage.

In decoding from pattern detection, the first stage essentially distinguishes a sequence obtained partially from the codewords and partially from the noise from a codeword sequence or a noise sequence. As a result, in the analysis of the probability of error, partial divergence and its properties described in Section~\ref{subsec-partial-divergence} play an important role. Using decoding from pattern detection, we have the following theorem.

\begin{theorem}
\label{th-achv2-basic}
Using decoding from pattern detection for intermittent communication $(\mathcal{X},\mathcal{Y},W, \star,\alpha)$, rates less than  $\max_{P}\{\left(\mathbb{I}(P,W)-f(P,W,\alpha)\right)^+\}$ are achievable, where
\begin{equation} \label{E: basic-f}
f(P,W,\alpha):=\max_{0 \le \beta \le \max\{1,\frac{1}{\alpha-1}\}} \{(\alpha-1)h(\beta)+ h((\alpha-1)\beta)-d_{(\alpha-1) \beta}(PW \| W_\star) -(\alpha-1)d_\beta(W_\star\|PW) \}.
\end{equation}
\end{theorem}

\begin{proof}
See Appendix~\ref{app-proof-th-achv2-basic}.
\end{proof}

\begin{remark}
\label{remark-achv-general}
The result of Theorem~\ref{th-achv2-basic} is valid for an arbitrary intermittent process in Figure~\ref{fig-system-model} provided $N/k \xrightarrow{\text{ } p \text{ }} \alpha$ as $k \to \infty$.
\end{remark}

\begin{remark}
\label{remark-interpretation}
The achievable rate in Theorem~\ref{th-achv2-basic} can be expressed as follows: Rate $R$ is achievable if for all mismatch $0 \le \beta \le \max\{1,\frac{1}{\alpha-1}\}$, we have
$$ R+(\alpha-1)h(\beta)+ h((\alpha-1)\beta) <  \mathbb{I}(P,W) + d_{(\alpha-1) \beta}(PW \| W_\star) +(\alpha-1)d_\beta(W_\star\|PW).$$
The interpretation is that the total amount of uncertainty should be smaller than the total amount of information. Specifically, $R$ and $(\alpha-1)h(\beta)+ h((\alpha-1)\beta)$ are the amount of uncertainty in codewords and patterns, respectively, and $\mathbb{I}(P,W)$ and $d_{(\alpha-1) \beta}(PW \| W_\star) +(\alpha-1)d_\beta(W_\star\|PW)$ are the amount of information about the codewords and patterns, respectively. The partial divergence terms are the total exponent in the probability of confusing the codeword / noise symbols under mismatch $\beta$, which is defined as the number of incorrectly chosen noise symbols divided by total number of insertions. For more details, please see Appendix~\ref{app-proof-th-achv2-basic}.
\end{remark}

The achievable rate in Theorem~\ref{th-achv2-basic} is always larger than the one in Theorem~\ref{theorem-simple-lower}, because decoding from pattern detection utilizes the fact that the choice of the codeword symbols at the receiver might not be a good one, and therefore, restricts the typicality decoding only to the typical patterns and decreases the search domain. In fact, the difference between the two achievable rates corresponds to the two partial divergence terms in~\eqref{E: basic-f}, which are the total exponent in the probability of confusing the codeword / noise symbols under mismatch $\beta$. The form of the achievable rate is reminiscent of communications overhead as the cost of constraints~\cite{overhead}, where the constraint is the system's burstiness or intermittency, and the overhead cost is $f(P,W,\alpha)$. Using the properties of partial divergence, we state some properties of this overhead cost in the next proposition.

\begin{proposition} \label{prop: overheadf-basic}
The overhead cost $f(P,W,\alpha)$ in~\eqref{E: basic-f} has the following properties:
\begin{enumerate}[(a)]
\item The maximum of the term in~\eqref{E: basic-f} occurs in the interval $[0,1/\alpha]$, i.e., instead of the maximization over $0 \le \beta \le \max\{1,\frac{1}{\alpha-1}\}$, $f(P,W,\alpha)$ can be found by the same maximization problem over $0 \le \beta \le 1/\alpha$. \label{prop3-a-basic}
\item $f(P,W,\alpha)$ is increasing in $\alpha$. \label{prop3-b-basic}
\item $f(P,W,1)=0$. \label{prop3-bb-basic}
\item If $D(PW\|W_\star)$ is finite, then $f(P,W,\alpha) \to \infty$ as $\alpha \to \infty$. \label{prop3-bbb-basic}
\end{enumerate}
\end{proposition}

\begin{proof}
See Appendix~\ref{app-proof-prop: overheadf-basic}.
\end{proof}

Note that part (\ref{prop3-b-basic}) in Proposition~\ref{prop: overheadf-basic} indicates that increasing the intermittency rate or the receive window increases the overhead cost, resulting in a smaller achievable rate. Parts (\ref{prop3-bb-basic}) and (\ref{prop3-bbb-basic}) show that the achievable rate is equal to the capacity of the channel for $\alpha=1$ and approaches zero as $\alpha \to \infty$.

\begin{remark}
The results in this section can be extended to packet-level intermittent communication in which the intermittency is modeled at the packet level, and regimes that lie between the non-contiguous transmission of codeword symbols in intermittent communication, and the contiguous transmission of codeword symbols in asynchronous communication are explored. Specifically, the system model in this paper can be generalized to small packet, medium packet, and large packet intermittent communication. See, for example,~\cite{mostafa-allerton-12}.
\end{remark}

Now consider a binary symmetric channel (BSC) for the DMC in Figure~\ref{fig-system-model} with the crossover probability $ 0 \le p < 0.5$, and the noise symbol $\star=0$. Figure~\ref{F: bsc-2} shows the value of the achievable rates for different $p$, versus $\alpha$. $R_{\text{Insertion}}$ denotes the achievable rate obtained from Theorem~\ref{th-achv2-basic} if the channel is noiseless ($p=0$), and can be proven to be equal to $\max_{0 \le p_0 \le 1} \{2h(p_0)-\max_{0 \le \beta \le 1} \{(\alpha-1) h(\beta)+h((\alpha-1)\beta)+(1-(\alpha-1)\beta)h(\frac{p_0-(\alpha-1)\beta}{1-(\alpha-1)\beta})\}\}$.

\begin{figure}[t]
\centerline{\scalebox{.38}{\includegraphics{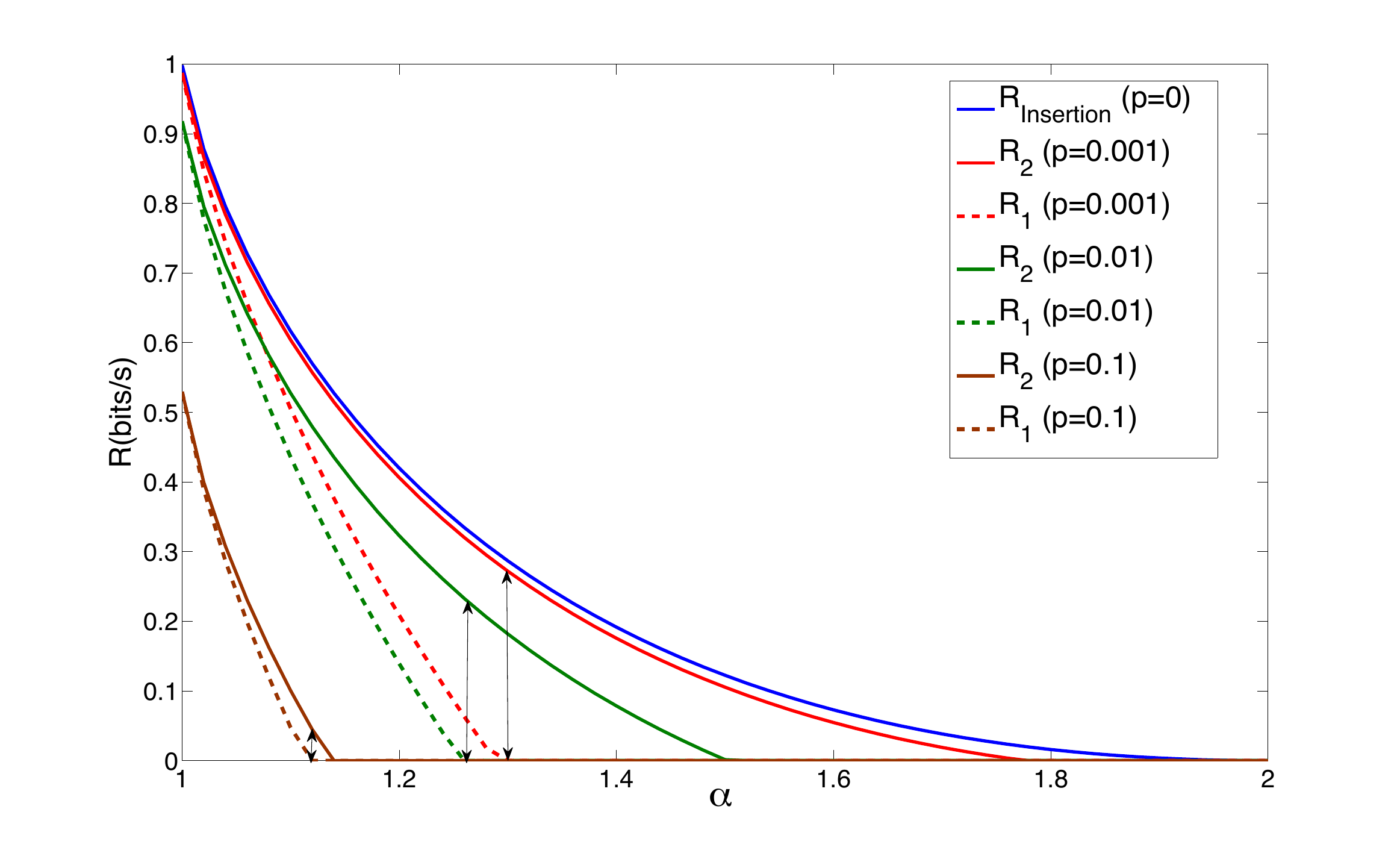}}}
\caption{Achievable rate region $(R,\alpha)$ for the BSC for different cross over probability $p$'s.}
\label{F: bsc-2}
\end{figure}

As we can see from the plot, the achievable rate in Theorem~\ref{th-achv2-basic} (indicated by ``$R_2$") is always larger than the one in Theorem~\ref{theorem-simple-lower} (indicated by ``$R_1$'') since decoding from pattern detection takes advantage of the fact that the choice of the $k$ output symbols might not be a good one. Specifically, the exponent obtained in Lemma~\ref{lemma1} in terms of the partial divergence helps the decoder detect the right symbols, and therefore, achieve a larger rate. The arrows in Figure~\ref{F: bsc-2} show this difference and suggest that the benefit of using decoding from pattern detection is larger for increasing $\alpha$. Note that the larger $\alpha$ is, the smaller the achievable rate would be for a fixed $p$. Not surprisingly, as $\alpha \to 1$, the capacity of the BSC is approached for both of the achievable rates. In this example, we cannot achieve a positive rate if $\alpha \ge 2$, even for the case of a noiseless channel ($p=0$). However, this is not true in general, because even the first achievable rate ($R_1$) can be positive for a large $\alpha$, if the capacity of the channel $C_W$ is sufficiently large. The results suggest that, as communication becomes more intermittent and $\alpha$ becomes larger, the achievable rate is decreased due to the additional uncertainty about the positions of the codeword symbols at the decoder.

\section{Upper Bounds}
\label{sec-upper}

In this section, we focus on obtaining upper bounds on the capacity of a special case of intermittent communication in which the DMC in Figure~\ref{fig-system-model} is binary-input binary-output noiseless with the noise symbol $\star=0$. The achievable rate for this case is denoted by $R_{\text{Insertion}}$ in Section~\ref{subsec-decode-algorithm}, and is shown by the blue curve in Figure~\ref{F: bsc-2}. Similar to~\cite{duman-outer-deletion1}, upper bounds are obtained by providing the encoder and the decoder with various amounts of side-information, and calculating or upper bounding the capacity of this genie-aided system. After introducing a useful function $g(a,b)$ in Section~\ref{subsec-auxiliary}, we obtain upper bounds in Section~\ref{subsec-genie}. The techniques of in this section can in principle be applied to non-binary and noisy channels as well; however, the computational complexity for numerical evaluation of the genie-aided system grows very rapidly in the size of the alphabets.

\subsection{Auxiliary Channel: Uniform Insertion}
\label{subsec-auxiliary}

Let $a$ and $b$ be two integer numbers such that $0 \le a \le b$, and consider a discrete memoryless channel for which at each channel use the input consists of a sequence of $a$ bits and the output consists of a sequence of $b$ bits, i.e., the input and output of this channel are $\mathbf{A} \in \{0,1\}^a$ and $\mathbf{B} \in \{0,1\}^b$, respectively. For each channel use, $b-a$ zeroes are inserted randomly and uniformly among the input symbols. The set of positions at which the insertions occur takes on each of the possible $\binom{b}{b-a}=\binom{b}{a}$  realizations with equal probability, and is unknown to the transmitter and the receiver. As an example, the transition probability matrix of this channel for the case of $a=2$ and $b=3$ is reported in Table~\ref{tab-1}.

\begin{table}[t]
\centering
\begin{tabular}{|l|c|c|c|c|c|c|c|c|}
\hline
\diaghead{\theadfont Diag kkkk}%
{$\mathbf{A}$}{$\mathbf{B}$}&
000 & 001 & 010 & 100 & 011 & 101 & 110 & 111\\
\hline
00 & 1 & 0 & 0 & 0 & 0 & 0 & 0 & 0\\
\hline
01 & 0 & 2/3 & 1/3 & 0 & 0 & 0 & 0 & 0\\
\hline
10 & 0 & 0 & 1/3 & 2/3 & 0 & 0 & 0 & 0\\
\hline
11 & 0 & 0 & 0 & 0 & 1/3 & 1/3 & 1/3 & 0\\
\hline
\end{tabular}
\caption{Transition probabilities $P(\mathbf{B}|\mathbf{A})$ for the auxiliary channel with $a=2$ and $b=3$.}
\label{tab-1}
\end{table}

The capacity of the auxiliary channel is defined as
\begin{equation} \label{e-func-g}
g(a,b):=\max_{P(\mathbf{A})} \mathbb{I}(\mathbf{A};\mathbf{B}), \quad 0 \le a \le b,
\end{equation}
where $P(\mathbf{A})$ is the input distribution. The exact value of the function $g(a,b)$ for finite $a$ and $b$ can be numerically computed by evaluating the transition probabilities $P(\mathbf{B}|\mathbf{A})$ and using the Blahut-Arimoto algorithm~\cite{blahut-cap} to maximize the mutual information. The largest value of $b$ for which we are able to evaluate the function $g(a,b)$ for all values of $a \le b$ is $b=17$. As defined in~\eqref{e-func-g}, $g(a,b)$ is the capacity of a DMC with input alphabet size of $2^a$ and output alphabet size of $2^b$ with a given probability transition matrix for the channel. The complexity of computing $g(a,b)$ using Blahut-Arimoto algorithm cannot be explicitly expressed as a function of $a$ and $b$, as it depends on the convergence rate of the algorithm, which in turn depends on the probability transition matrix. However, the complexity of one iteration of the Blahut-Arimoto algorithm is equal to $O(2^{a+b})$.

Although we cannot obtain a closed-form expression for the function $g(a,b)$, we can find a closed-form upper and lower bound by expanding the mutual information in~\eqref{e-func-g} and bounding some of its terms. The upper and lower bound and a numerical comparison between the lower bound, the exact value, and the upper bound on the function $g(a,b)$ can be found in~\cite{mostafa-ita-13}.

The following definition will be useful in expressing the upper bounds in Section~\ref{subsec-genie}.
\begin{equation} \label{e-auxiliary-cap2}
\phi(a,b):=a-g(a,b),
\end{equation}
Note that the function $\phi(a,b)$ quantifies the loss in capacity due to the uncertainty about the positions of the insertions, and cannot be negative. The following proposition characterizes some of the properties of the functions $g(a,b)$ and $\phi(a,b)$, which will be used later.
\begin{proposition} \label{prop-auxil1}
The functions $g(a,b)$ and $\phi(a,b)$ have the following properties:
\begin{enumerate}[(a)]
\item $g(a,b) \le a$, $\phi(a,b) \ge 0$. \label{prop-aux-a}
\item $g(a,a)=a$, $\phi(a,a)=0$. \label{prop-aux-b}
\item $g(1,b)=1$, $\phi(1,b)=0$. \label{prop-aux-c}
\item $g(a,b+1) \le g(a,b)$, $\phi(a,b+1) \ge \phi(a,b)$. \label{prop-aux-d}
\item $g(a+1,b+1) \le 1+g(a,b)$, $\phi(a+1,b+1) \ge \phi(a,b)$. \label{prop-aux-e}
\end{enumerate}
\end{proposition}

\begin{proof}
See Appendix~\ref{app-proof-prop-auxil1}.
\end{proof}

\subsection{Genie-Aided System and Numerical Upper Bounds}
\label{subsec-genie}

In this section, we focus on upper bounds on the capacity of binary-input binary-output noiseless intermittent communication. The procedure is similar to~\cite{duman-outer-deletion1}. Specifically, we obtain upper bounds by giving some kind of side-information to the encoder and decoder, and calculating or upper bounding the capacity of this genie-aided channel.

Now we introduce one form of side-information. Assume that the position of the $[(s+1)i]^{th}$ codeword symbol in the output sequence is given to the encoder and decoder for all $i=1,2,...$ and a fixed integer number $s \ge 1$. We assume that the codeword length is a multiple of $s+1$, so that $t=k/(s+1)$ is an integer, and is equal to the total number of positions that are provided as side-information. This assumption does not impact the asymptotic behavior of the channel as $k \to \infty$. We define the random sequence $\{Z_i \}_{i=1}^t$ as follows: $Z_1$ is equal to the position of the $[s+1]^{th}$ codeword symbol in the output sequence, and for $i \in \{2,3,...,t\}$, $Z_i$ is equal to the difference between the positions of the $[(s+1)i]^{th}$ codeword symbol and $[(s+1)(i-1)]^{th}$ codeword symbol in the output sequence. Figure~\ref{F: outer-Zi} illustrate an example of the variable $Z_i$ for the case of $s=2$, where $c_i$ denotes the $i^{th}$ codeword symbol and $0$'s are the inserted noise symbols. 

\begin{figure}[t]
\centerline{\scalebox{.28}{\includegraphics{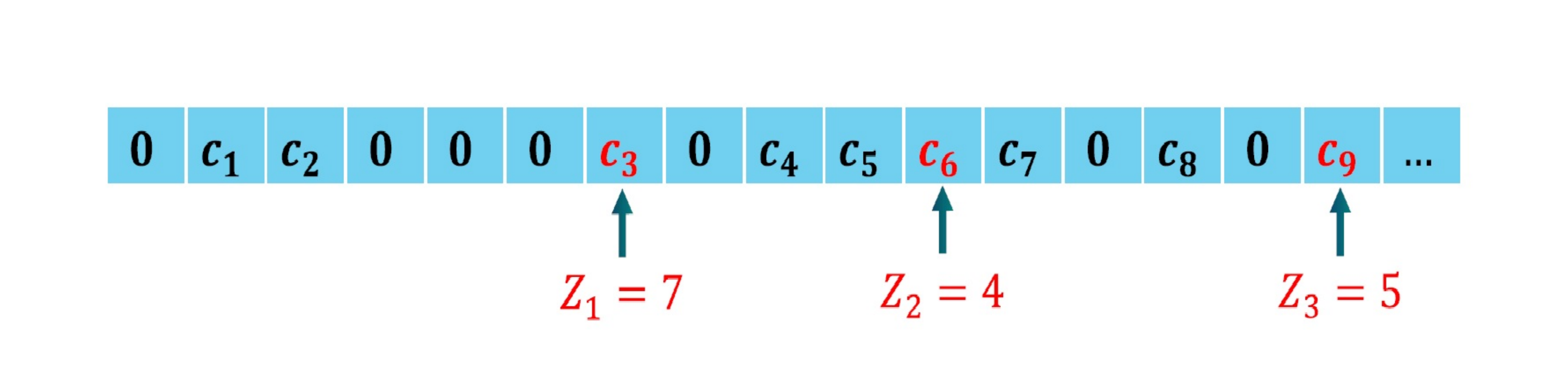}}}
\caption{An example of the variable $Z_i$ for the case of $s=2$, where $c_i$ denotes the $i^{th}$ codeword symbol.}
\label{F: outer-Zi}
\end{figure}

Since we assumed iid insertions, the random sequence $\{Z_i \}_{i=1}^t$ is iid as well with distribution:
\begin{equation} \label{distr-z-1}
P(Z_i=b+1)=\binom{b}{s}(1-p_t)^{b-s} p_t^{s+1}, b \ge s,
\end{equation}
with mean $\E {Z_i}=(s+1)/p_t$. Also, note that as $k \to \infty$, by the law of large numbers, we have
\begin{equation} \label{e-zz}
\frac{N}{t} \xrightarrow{\text{ } p \text{ }} \E {Z_i}=\frac{s+1}{p_t}.
\end{equation}

Let $C_1$ denote the capacity of the channel if we provide the encoder and decoder with side-information on the random sequence $\{Z_i \}_{i=1}^t$, which is clearly an upper bound on the capacity of the original channel. With this side-information, we essentially partition the transmitted and received sequences into $t$ contiguous blocks that are independent from each other. In the $i^{th}$ block the place of the $[(s+1)i]^{th}$ codeword symbol is given, which can convey one bit of information. Other than that, the $i^{th}$ block has $s$ input bits and $Z_i-1$ output bits with uniform $0$ insertions. Therefore, the information that can be conveyed through the $i^{th}$ block equals $g(s,Z_i-1)+1$. Thus, we have
\begin{align}
C_1=& \lim_{k \to \infty} \frac{1}{k} \sum_{i=1}^{t} g(s,Z_i-1)+1 \notag \\
=& \lim_{k \to \infty} \frac{N}{k} \frac{t}{N} \frac{1}{t} \sum_{i=1}^{t} g(s,Z_i-1)+1 \notag \\
=& \frac{1}{s+1} \lim_{t \to \infty} \frac{1}{t} \sum_{i=1}^{t} g(s,Z_i-1)+1 \label{cap1-h1} \\
=& \frac{1}{s+1} \E{g(s,Z_i-1)+1} \label{cap1-h2} \\
=& \frac{1}{s+1} \left[ 1+\sum_{b=s}^{\infty} \binom{b}{s}(1-p_t)^{b-s} p_t^{s+1} g(s,b) \right] \label{cap1-h3} \\
=& 1- \frac{1}{s+1} \sum_{b=s}^{\infty} \binom{b}{s}(1-p_t)^{b-s} p_t^{s+1} \phi(s,b), \label{cap1-h4}
\end{align}
where:~\eqref{cap1-h1} follows from~\eqref{e-zz};~\eqref{cap1-h2} follows from the law of large numbers;~\eqref{cap1-h3} follows from the distribution of $Z_i$'s given in~\eqref{distr-z-1}; and~\eqref{cap1-h4} follows from the definition~\eqref{e-auxiliary-cap2}. Note that the capacity $C_1$ cannot be larger than $1$, since the coefficients $\phi(\cdot,\cdot)$ cannot be negative. The negative term in~\eqref{cap1-h4} can be interpreted as a lower bound on the communication overhead as the cost of intermittency in the context of~\cite{overhead}.

The expression in~\eqref{cap1-h4} gives an upper bound on the capacity of the original channel with $p_t=1/\alpha$. However, it is infeasible to numerically evaluate the coefficients $\phi(s,b)$ for large values of $b$. As we discussed before, the largest value of $b$ for which we are able to evaluate the function $\phi(s,b)$ is $b_{max}=17$. The following upper bound on $C_1$ results by truncating the summation in~\eqref{cap1-h4} and using part (\ref{prop-aux-d}) of Proposition~\ref{prop-auxil1}.
\begin{equation}
C_1 \! \le 1- \frac{\phi(s,b_{max})}{s+1}  \!+\! \frac{1}{s\!+\!1}\!\sum_{b=s}^{b_{max}}\! \binom{b}{s} p_t^{s+1}(1\!-\!p_t)^{b-s} (\phi(s,b_{max})\!-\!\phi(s,b)), \label{cap2-h2}
\end{equation}
The expression~\eqref{cap2-h2}, which we denote by $C_1^\prime$, gives a nontrivial and computable upper bound for each value of $s=2,3,...,b_{max}-1$ on $C_1$, and therefore, an upper bound on the capacity of the original channel with $p_t=1/\alpha$. Figure~\ref{F: outer-2} shows the upper bounds for $b_{max}=17$ and $s=2,3,...,16$ versus the intermittency rate $\alpha$, along with the the achievability result.

\begin{figure}[t]
\centerline{\scalebox{.38}{\includegraphics{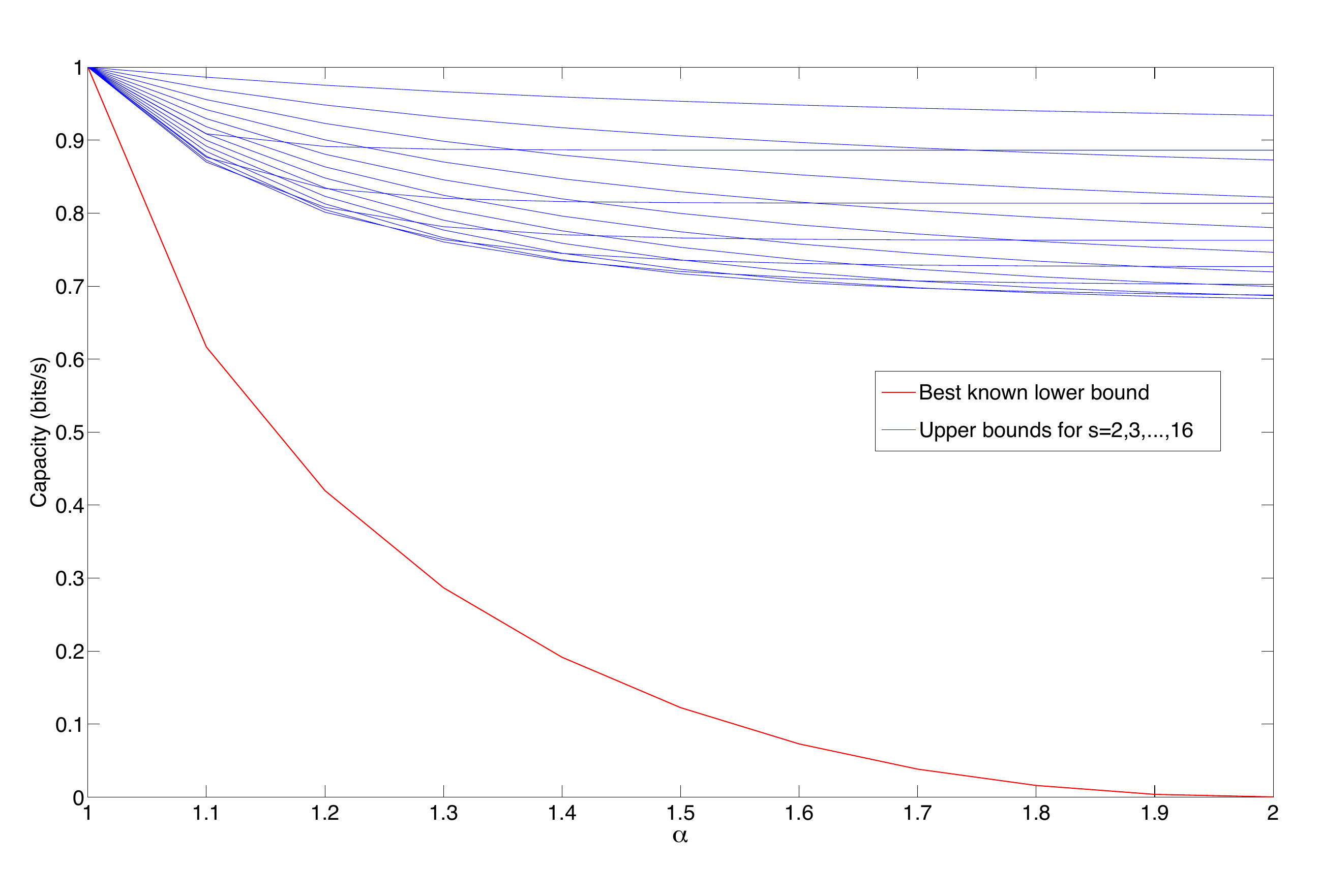}}}
\caption{Comparison between the best achievability result with different upper bounds obtained from~\eqref{cap2-h2} for $b_{max}=17$ and $s=2,3,...,16$, versus the intermittency rate $\alpha$.}
\label{F: outer-2}
\end{figure}

Next, we introduce a second form of side-information. Assume that for consecutive blocks of length $s$ of the output sequence, the number of codeword symbols within that block is given to the encoder and decoder as side-information, i.e., the number of codeword symbols in the sequence $(y_{(i-1)s+1} ,y_{(i-1)s+2},...,y_{is})$, $i=1,2,...$ for a fixed integer number $s \ge 2$. Let $C_2$ denote the capacity of the channel if we provide the encoder and decoder with this side-information. Using a similar procedure, we obtain
\begin{equation} \label{cap3}
C_2=1-\frac{1}{sp_t} \sum_{a=0}^s \binom{s}{a}p_t^a(1-p_t)^{s-a} \phi(a,s).
\end{equation}
Note that the summation in~\eqref{cap3} is finite, and we do not need to upper bound $C_2$ as we did for $C_1$. The value of $C_2$ gives nontrivial and computable upper bounds on the capacity of the original channel. Figure~\ref{F: outer-3} shows the upper bounds for $s=3,4,...,17$ versus the intermittency rate $\alpha$, along with the the achievability result. The upper bound corresponding to $s=17$ is tighter than others for all ranges of $\alpha$, i.e.,~\eqref{cap3} is decreasing in $s$. Intuitively, this is because by decreasing $s$, we provide the side-information more frequently, and therefore, the capacity of the resulting genie-aided system becomes larger.

\begin{figure}[t]
\centerline{\scalebox{.38}{\includegraphics{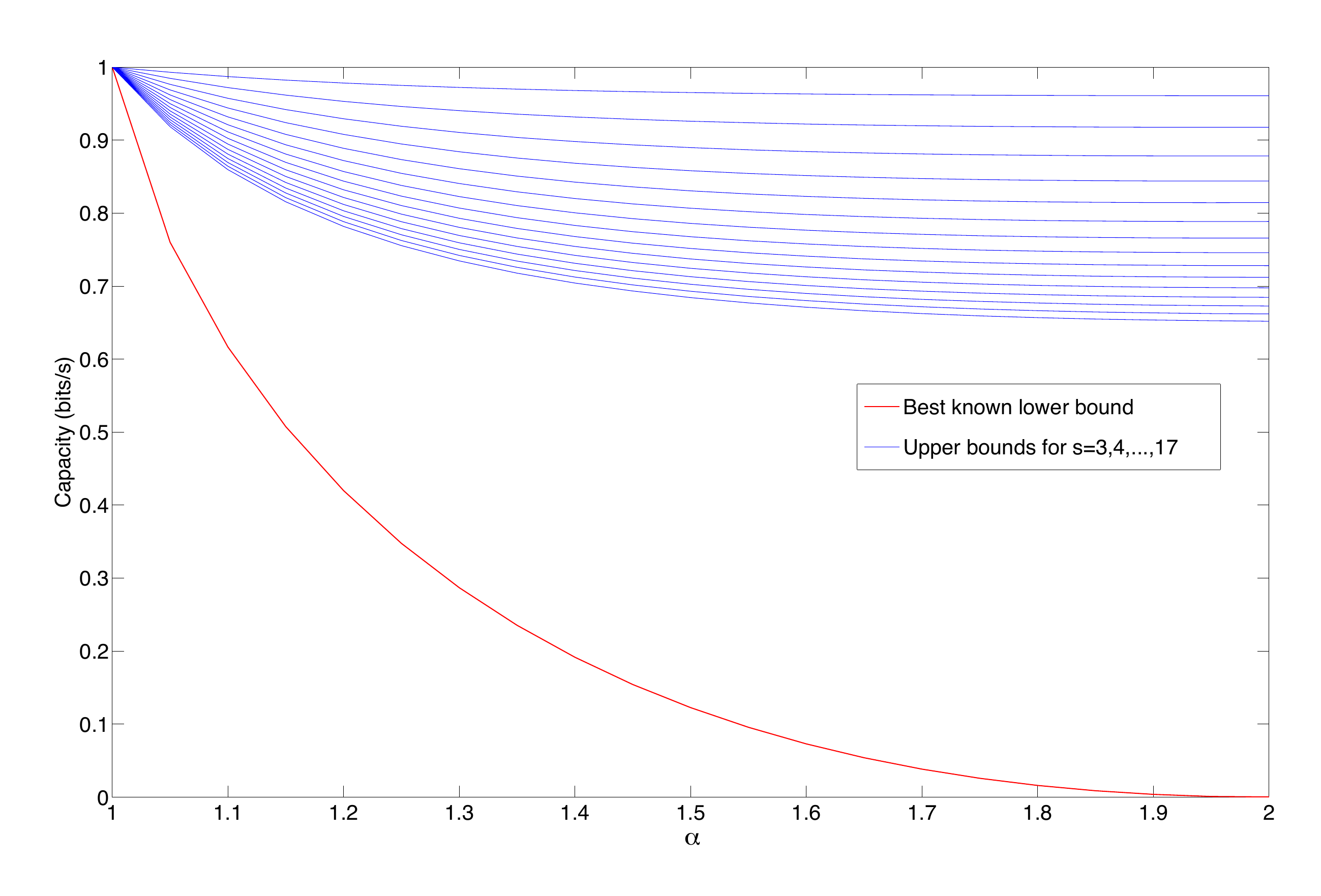}}}
\caption{Comparison between the best achievability result with different upper bounds obtained from~\eqref{cap3} for $s=3,4,...,17$, versus the intermittency rate $\alpha$.}
\label{F: outer-3}
\end{figure}

It seems that~\eqref{cap3} gives better upper bounds for the range of $\alpha$ shown in the figures ($1 < \alpha \le 2$). However, the other upper bound $C_1^\prime$ can give better results for the limiting values of $\alpha \to \infty$ or $p_t \to 0$. We have
\begin{align}
& \lim_{\alpha \to \infty} C_1^\prime=1-\frac{\phi(s,b_{max})}{s+1}, \label{e-lim-alpha-inf} \\
& \lim_{\alpha \to \infty} C_2=1. \notag
\end{align}
This is because of the fact that by increasing $\alpha$, and thus decreasing $p_t$, we have more zero insertions and the first kind of genie-aided system provides side-information less frequently leading to tighter upper bounds. The best upper bound for the limiting case of $\alpha \to \infty$ found by~\eqref{e-lim-alpha-inf} is $0.6307$ bits per channel use.

Although the gap between the achievable rates and upper bounds is not particularly tight, especially for large values of intermittency rate $\alpha$, the upper bounds suggest that the linear scaling of the receive window with respect to the codeword length considered in the system model is natural since there is a tradeoff between the capacity of the channel and the intermittency rate. By contrast, in asynchronous communication~\cite{asynch-wornell-1, asynch-wornell}, where the transmission of the codeword is contiguous, only exponential scaling $n=e^{\alpha k}$ induces a tradeoff between capacity and asynchronism.

\section{Conclusion}
\label{sec: conclusion}
We formulated a model for intermittent communication that can capture bursty transmissions or a sporadically available channel by inserting a random number of silent symbols between each codeword symbol so that the receiver does not know a priori when the transmissions will occur. We introduced decoding from pattern detection in order to develop achievable rates using partial divergence and its properties. As the system becomes more intermittent, the achievable rates decrease due to the additional uncertainty about the positions of the codeword symbols at the decoder. For the case of binary-input binary-output noiseless channel, we obtained upper bounds on the capacity of intermittent communication by providing the encoder and the decoder with various amounts of side-information, and calculating or upper bounding the capacity of this genie-aided system. The results suggest that the linear scaling of the receive window with respect to the codeword length considered in the system model is relevant since the upper bounds imply a tradeoff between the capacity and the intermittency rate.

\appendices

\section{Proof of Proposition~\ref{prop1}}
\label{app-proof-prop1}

It can readily be verified that the function $d(P,Q,Q^\prime,\rho)$ in~\eqref{eq-partial-lemma} can be written as~\cite{mostafa-isit-12}
\begin{equation}
d(P,Q,Q^\prime,\rho)=H(P)+D(P\|Q) -\bar{\rho}\log\frac{q_{|\mathcal{X}|}^\prime}{q_{|\mathcal{X}|}}-e(P,Q,Q^\prime,\rho), \label{E: d}
\end{equation}
where
\begin{align}
&e(P,Q,Q^\prime,\rho):=\max_{0 \le \theta_j \le 1, j=1,2,...,|\mathcal{X}|-1} \{\rho H(P_1)+\bar{\rho} H(P_2) +\sum_{j=1}^{|\mathcal{X}|-1}\theta_j p_j\log a_j \}, \label{E: e} \\
& a_j:=\frac{q_j^\prime q_{|\mathcal{X}|}}{q_j q_{|\mathcal{X}|}^\prime}, \quad j=1,2,...,|\mathcal{X}|-1, \label{E: aj} \\
&P_1:=(\frac{\bar{\theta}_1 p_1}{\rho},\frac{\bar{\theta}_2 p_2}{\rho},...,\frac{\bar{\theta}_{|\mathcal{X}|-1} p_{|\mathcal{X}|-1}}{\rho},1-\frac{\sum_{j=1}^{|\mathcal{X}|-1} \bar{\theta}_j p_j}{\rho}), \label{E: p1} \\
&P_2:=(\frac{\theta_1 p_1}{\bar{\rho}},\frac{\theta_2 p_2}{\bar{\rho}},...,\frac{\theta_{|\mathcal{X}|-1} p_{|\mathcal{X}|-1}}{\bar{\rho}},1-\frac{\sum_{j=1}^{|\mathcal{X}|-1} \theta_j p_j}{\bar{\rho}}). \label{E: p2}
\end{align}

Now, we prove the proposition by solving the optimization problem in~\eqref{E: e} and simplifying~\eqref{E: d} for the special case of $Q^\prime=P$. We would like to find the optimal $\Theta:=(\theta_1,\theta_2,...,\theta_{|\mathcal{X}|-1})$ for which the function $$\tilde{e}(P,Q,\rho,\Theta):=\rho H(P_1)+\bar{\rho} H(P_2)+\sum_{j=1}^{|\mathcal{X}|-1}\theta_j p_j\log a_j$$ is maximized subject to the constraints
\begin{align}
&0 \le \theta_j \le 1, \quad j=1,2,...,|\mathcal{X}|-1, \label{E: constraint1} \\
&\sum_{j=1}^{|\mathcal{X}|-1} \bar{\theta}_j p_j \le \rho, \label{E: constrain2} \\
&\sum_{j=1}^{|\mathcal{X}|-1} \theta_j p_j \le \bar{\rho}, \label{E: constrain3}
\end{align}
where~\eqref{E: constrain2} and~\eqref{E: constrain3} arise from the fact that $P_1$ and $P_2$ should be probability mass functions, respectively. Note that $\tilde{e}$ is concave in $\Theta$, because $H(P)$ is a concave function in $P$, and $P_1$ and $P_2$ are linear functions in $\Theta$. Therefore, the optimal $\Theta$ can be found by setting the derivative of $\tilde{e}$ with respect to $\theta_j$'s, $j=1,2,...,|\mathcal{X}|-1$ to zero, and verifying that the solution satisfies the constraints~\cite{boyd2004convex}. We have
\begin{equation} \notag
\frac{\partial \tilde{e}}{\partial \theta_j}=p_j \log \frac{\theta_j(\rho-\sum_{j=1}^{|\mathcal{X}|-1} \bar{\theta}_j p_j)}{\bar{\theta}_j(\bar{\rho}-\sum_{j=1}^{|\mathcal{X}|-1} \theta_j p_j) a_j}=0,  \quad j=1,2,...,|\mathcal{X}|-1.
\end{equation}
Therefore, we have
\begin{equation} \label{E: set der zero}
\frac{\bar{\theta}_j}{\theta_j} \frac{p_j}{q_j}=c^*=\frac{p_{|\mathcal{X}|}}{q_{|\mathcal{X}|-1}}\frac{\rho-\sum_{j=1}^{|\mathcal{X}|-1} \bar{\theta}_j p_j}{\bar{\rho}-\sum_{j=1}^{|\mathcal{X}|-1} \theta_j p_j}, \quad j=1,2,...,|\mathcal{X}|-1,
\end{equation}
where $c^*$ is defined to be equal to the right term in~\eqref{E: set der zero} since it is fixed for all $j$'s. From the left equality in~\eqref{E: set der zero}, the optimal solution is characterized by
\begin{equation} \label{E: opt theta}
\theta_j^*=\frac{p_j}{c^* q_j+p_j}, \quad j=1,2,...,|\mathcal{X}|-1,
\end{equation}
and from the right equality in~\eqref{E: set der zero}, $c^*$ is found to be the solution to~\eqref{E: const c}. Now, we check that this solution satisfies the constraints. Let $g(c):=c \sum_{j=1}^{|\mathcal{X}|} \frac{p_j q_j}{c q_j+p_j}-\rho$ be a function of $c$ for a fixed $P$, $Q$, and $\rho$ whose root gives the value of $c^*$. Note that the equation $g(c)=0$ cannot have more than one solutions, otherwise, the concave and differentiable function $\tilde{e}$ would have more than one local maximums, which is a contradiction. Also, note that
\begin{align}
\lim_{c \to 0} g(c)&=-\rho \le 0, \notag \\
\lim_{c \to \infty} g(c)&=1-\rho \ge 0, \notag
\end{align}
and because the function $g(c)$ is continuous in $c$, the root of this function is unique and is non-negative. Therefore, $c^* \ge 0$ is the unique solution to~\eqref{E: const c}. From~\eqref{E: opt theta} and the fact that $c^* \ge 0$, the constraint~\eqref{E: constraint1} is satisfied. From the right equality in~\eqref{E: set der zero} and the fact that $c^* \ge 0$, we observe that
$$\frac{\rho-\sum_{j=1}^{|\mathcal{X}|-1} \bar{\theta}_j p_j}{\bar{\rho}-\sum_{j=1}^{|\mathcal{X}|-1} \theta_j p_j} \ge 0.$$
Also, note that numerator and denominator of the above expression cannot simultaneously be negative, and therefore, both constraints~\eqref{E: constrain2} and~\eqref{E: constrain3} are satisfied as well.

Using the optimal solution obtained above, the elements of the PMF's $P_1$ and $P_2$ in~\eqref{E: p1} and~\eqref{E: p2}, respectively, are obtained as
\begin{align}
&p_{1,j}=\frac{p_j q_j}{c^*q_j+p_j}/\sum_{j=1}^{|\mathcal{X}|} \frac{p_j q_j}{c^*q_j+p_j}, \notag \\
&p_{2,j}=\frac{p_j^2}{c^*q_j+p_j}/\sum_{j=1}^{|\mathcal{X}|} \frac{p_j^2}{c^*q_j+p_j}, \notag
\end{align}
for $j=1,2,...,|\mathcal{X}|$. Substituting the optimal solution into~\eqref{E: e} and then into~\eqref{E: d},~\eqref{E: explicit d} can be obtained after some manipulation.

\section{Proof of Proposition~\ref{prop2}}
\label{app-proof-prop2}

 \begin{enumerate}[(a)]

 \item From~\eqref{E: const c}, we observe that $c^* \to 0$ as $\rho \to 0$. Therefore, from~\eqref{E: explicit d}, we have
 \begin{equation} \notag
d_0(P\|Q)=D(P\|Q)-\sum_{j=1}^{|\mathcal{X}|}p_j \log \frac{p_j}{q_j} +c^* \log c^*+h(0)=0,
\end{equation}
for $c^* \to 0$.

 \item From~\eqref{E: const c}, we observe that $c^* \to \infty$ as $\rho \to 1$. Therefore, from~\eqref{E: explicit d}, we have
 \begin{equation} \notag
\!\!\!\!\!d_1(P\|Q)\!=\!D(P\|Q)-\sum_{j=1}^{|\mathcal{X}|}p_j \log (1+\frac{p_j}{c^*q_j}) +h(1)\!=\!D(P\|Q),
\end{equation}
for $c^* \to \infty$.

 \item If $P=Q$, then~\eqref{E: const c} simplifies to $\frac{c^*}{c^*+1}=\rho$, and therefore $c^*=\rho/\bar{\rho}$. By substituting $c^*$ into~\eqref{E: explicit d} and because $D(P\|P)=0$, we obtain
 \begin{equation} \notag
 d_\rho(P\|P)=-\log\left(\frac{\rho}{\bar{\rho}}+1\right)+\rho \log \frac{\rho}{\bar{\rho}} +h(\rho)=0.
 \end{equation}

 \item By taking the derivative of~\eqref{E: explicit d} with respect to $\rho$, we obtain
 \begin{align}
\!\!d_\rho^\prime(P\|Q)\!\!&=\!-\frac{\partial c^*}{\partial \rho} \sum_{j=1}^{|\mathcal{X}|} \frac{p_j q_j}{c^*q_j+p_j}+\log c^*+\frac{\partial c^*}{\partial \rho} \frac{\rho}{c^*}+\log\frac{\bar{\rho}}{\rho} \notag \\
& = \frac{\partial c^*}{\partial \rho} (\frac{\rho}{c^*} - \sum_{j=1}^{|\mathcal{X}|} \frac{p_j q_j}{c^*q_j+p_j})+\log (c^* \frac{\bar{\rho}}{\rho}) \notag \\
&= \log (c^* \frac{\bar{\rho}}{\rho}), \label{E: derv partial diver}
 \end{align}
 where~\eqref{E: derv partial diver} is obtained by using~\eqref{E: const c}. Therefore,
 \begin{equation} \notag
 d_0^\prime(P\|Q)=\lim_{\rho \to 0} d_\rho^\prime(P\|Q)=\lim_{\rho \to 0} \log (c^* \frac{\bar{\rho}}{\rho})=0,
\end{equation}
because we have $\lim_{\rho \to 0} \frac{\rho}{c^*}=\sum_{j=1}^{|\mathcal{X}|}\frac{p_j q_j}{p_j}=1$ from~\eqref{E: const c}.

\item According to~\eqref{E: derv partial diver}, in order to prove $d_\rho^\prime(P\|Q) > 0,  0 < \rho \le 1$, it is enough to show that $\frac{\rho}{\bar{\rho}} < c^*, 0 < \rho \le 1$:

\begin{align}
1&=\left(\sum_{j=1}^{|\mathcal{X}|} p_j\right)^2 \notag \\
&< \sum_{j=1}^{|\mathcal{X}|} (\sqrt{c^*q_j+p_j})^2 \cdot \sum_{j=1}^{|\mathcal{X}|} (\frac{p_j}{\sqrt{c^*q_j+p_j}})^2 \label{E: cauchy-1} \\
&=(c^*+1) \sum_{j=1}^{|\mathcal{X}|} \frac{p_j^2}{c^*q_j+p_j} \notag \\
&=(c^*+1) \bar{\rho}, \label{E: bar rho}
\end{align}
where~\eqref{E: cauchy-1} follows from the Cauchy-Schwarz inequality, and~\eqref{E: bar rho} is true because
\begin{align}
\bar{\rho}&=1-\rho=\sum_{j=1}^{|\mathcal{X}|} p_j - \rho \notag \\
&= \sum_{j=1}^{|\mathcal{X}|} p_j-\sum_{j=1}^{|\mathcal{X}|} \frac{c^* p_j q_j}{c^* q_j+p_j} \label{E: ineq111} \\
&=\sum_{j=1}^{|\mathcal{X}|} \frac{p_j^2}{c^*q_j+p_j}, \notag
\end{align}
where~\eqref{E: ineq111} follows from~\eqref{E: const c}. Note that the Cauchy-Schwarz inequality in~\eqref{E: cauchy-1} cannot hold with equality for $0 <\rho$ (and therefore, for $0<c^*$), because otherwise, $p_j=q_j, j=1,2,...,|\mathcal{X}|$, and $P=Q$. From~\eqref{E: bar rho}, $1 < (c^*+1) \bar{\rho}$, which results in the desirable inequality $\frac{\rho}{\bar{\rho}} < c^*$.

\item By taking the derivative of~\eqref{E: derv partial diver} with respect to $\rho$, it can be seen that
\begin{equation} \label{E: sec derv partial}
d_\rho^{\prime\prime}(P\|Q)=\frac{1}{c^*}\frac{\partial c^*}{\partial \rho} -\frac{1}{\rho \bar{\rho}}.
\end{equation}
Also, by taking the derivative of~\eqref{E: const c} with respect to $\rho$ and after some calculation, we have
\begin{equation} \label{E: derv of c}
\sum_{j=1}^{|\mathcal{X}|}\frac{(c^*q_j)^2 p_j}{(c^*q_j+p_j)^2}=\rho-\frac{c^*}{\frac{\partial c^*}{\partial \rho}}.
\end{equation}
Therefore,
\begin{align}
\rho^2&=\left( \sum_{j=1}^{|\mathcal{X}|} \frac{c^* p_j q_j}{c^* q_j+p_j} \right)^2 \notag \\
&<\sum_{j=1}^{|\mathcal{X}|}\frac{(c^*q_j)^2 p_j}{(c^*q_j+p_j)^2} \cdot \sum_{j=1}^{|\mathcal{X}|} p_j \label{E: cauchy-2} \\
&=\rho-\frac{c^*}{\frac{\partial c^*}{\partial \rho}}, \label{E: to get convx}
\end{align}
where~\eqref{E: cauchy-2} follows from the Cauchy-Schwarz inequality, which cannot hold with equality since otherwise $P=Q$, and where~\eqref{E: to get convx} follows from~\eqref{E: derv of c}. From~\eqref{E: to get convx}, $\rho^2 < \rho-\frac{c^*}{\frac{\partial c^*}{\partial \rho}}$, which implies that $d_\rho^{\prime\prime}(P\|Q)>0$ according to~\eqref{E: sec derv partial}.

\item From part~\eqref{e}, $d_\rho(P\|Q)$ is convex in $\rho$, and therefore, $\frac{d_\rho(P\|Q)}{\rho}$ is increasing in $\rho$. In addition, from part~\eqref{c}, $\lim_{\rho \to 0} \frac{d_\rho(P\|Q)}{\rho}=d_0^\prime(P\|Q)=0$, and from part~\eqref{b}, $\lim_{\rho \to 1} \frac{d_\rho(P\|Q)}{\rho}=D(P\|Q)$. Consequently, $0 \le d_\rho(P\|Q) \le \rho D(P\|Q)$.

 \end{enumerate}

\section{Proof of Theorem~\ref{th-achv2-basic}}
\label{app-proof-th-achv2-basic}

Fix the input distribution $P$, and consider decoding from pattern detection described in Section~\ref{subsec-decode-algorithm}. For any $\epsilon >0$, we prove that if $R=\mathbb{I}(P,W)-f(P,W,\alpha)-2\epsilon$, then the average probability of error, $p_e^{avg}$, vanishes as $k \to \infty$. We have
\begin{equation} \label{E: union1}
p_e^{avg} = \mathbb{P} (e|m=1) \le  \mathbb{P}(\hat{m} \in \{2,3,...,M\} |m=1) + \mathbb{P}(\hat{m}=e|m=1),
\end{equation}
where the equality in~\eqref{E: union1} follows from~\cite[(7.70)]{Cover}, and the inequality follows from the union bound in which the last term is the probability that the decoder declares an error (does not find any message) at the end of all $\binom{n}{k}$ choices, which implies that even if we pick the correct output symbols, the decoder either does not pass the first stage or does not declare $m=1$ in the second stage. Therefore,
\begin{align}
\mathbb{P}(\hat{m}=e|m=1) \le & \mathbb{P}(Y^k \notin T_{[PW]_\mu}) +\mathbb{P}(Y_\star^{n-k} \notin T_{[W_\star]_\mu}) +\mathbb{P}(Y^k \notin T_{[W]_\mu}(C^k(1))) \label{E: secterm} \\
& \to 0, \text{ as } k \to \infty, \label{E: secterm1}
\end{align}
where $Y^k$ is the output of the channel if the input is $C^k(1)$, and $Y_\star$ is the output of the channel if the input is the noise symbol, and where we use the union bound to establish~\eqref{E: secterm}. The limit~\eqref{E: secterm1} follows because all the three terms in~\eqref{E: secterm} vanish as $k \to \infty$ according to Lemma~\ref{fact-typ2}.

The first term after the inequality in~\eqref{E: union1} is more challenging. It is the probability that for at least one choice of the output symbols, the decoder passes the first stage and then the typicality decoder declares an incorrect message. We characterize the $\binom{n}{k}$ choices based on the number of incorrectly chosen output symbols, i.e., the number of symbols in $\tilde{y}^k$ that are in fact output symbols corresponding to a noise symbol, which is equal to the number of symbols in $\hat{y}^{n-k}$ that are in fact output symbols corresponding to a codeword symbol. For any $0 \le k_1 \le \max\{k,n-k\}$, there are $\binom{k}{k_1} \binom{n-k}{k_1}$ choices. \footnote{According to Vandermonde's identity, we have $ \mathsmaller{\sum_{k_1=0}^{\max\{k,n-k\}} \binom{k}{k_1} \binom{n-k}{k_1} = \binom{n}{k}}$.} Using the union bound for all the choices and all the messages $\hat{m} \ne 1$, we have
\begin{equation}
\mathbb{P}(\hat{m} \in  \{2,3,...,M\}  |m=1) \le (e^{kR}-1) \sum_{k_1=0}^{\max\{k,n-k\}} \binom{k}{k_1} \binom{n-k}{k_1} \mathbb{P}_{k_1}(\hat{m} =2 |m=1), \label{E: eq1}
\end{equation}
where $\mathbb{P}_{k_1}(\hat{m} =2 |m=1)$ in~\eqref{E: eq1} denotes the probability that the decoder declares $\hat{m} =2$ conditioned on the fact that message $m=1$ is sent and the number of incorrectly chosen output symbols in decoding from pattern detection described in Section~\ref{subsec-decode-algorithm} is equal to $k_1$. Note that message $\hat{m} =2$ is declared at the decoder only if it passes the first and the second stage. Therefore,
\begin{align}
&\mathbb{P}_{k_1}(\hat{m} =2 | m=1) \notag \\
&=\mathbb{P}_{k_1} \bigg(\{\tilde{Y}^k \in T_{[PW]_\mu}\}  \cap \{\hat{Y}^{n-k} \in T_{[W_\star]_\mu}\} \cap \{\tilde{Y}^k \in T_{[W]_\mu}(C^k(2))\}|m=1\bigg) \notag \\
&=\mathbb{P}_{k_1}(\tilde{Y}^k \in T_{[PW]_\mu})\cdot \mathbb{P}_{k_1}(\hat{Y}^{n-k} \in T_{[W_\star]_\mu}) \mathsmaller{\cdot \mathbb{P}(\tilde{Y}^k \in T_{[W]_\mu}(C^k(2))|m\!=\!1,\! \tilde{Y}^k \! \in \! T_{[PW]_\mu},\! \hat{Y}^{n-k} \! \in \!T_{[W_\star]_\mu})} \label{E: eq2} \\
& \le e^{o(k)}e^{-k d_{k_1 /k}(PW\|W_\star)}e^{-(n-k) d_{k_1 /(n-k)}(W_\star\|PW)}  e^{-k(\mathbb{I}(P,W)-\epsilon)}, \label{E: eq3}
\end{align}
where:~\eqref{E: eq2} follows from the independence of the events $\{\tilde{Y}^k \in T_{[PW]_\mu}\}$ and $\{\hat{Y}^{n-k} \in T_{[W_\star]_\mu}\}$ conditioned on $k_1$ incorrectly chosen output symbols; and~\eqref{E: eq3} follows from using Lemma~\ref{lemma1} for the first two terms in~\eqref{E: eq2} with mismatch ratios $k_1 /k$ and $k_1 /(n-k)$, respectively, and using Lemma~\ref{fact1} for the last term in~\eqref{E: eq2}, because conditioned on message $m=1$ being sent, $C^k(2)$ and $\tilde{Y}^k$ are independent regardless of the other conditions in the last term. Substituting~\eqref{E: eq3} into the summation in~\eqref{E: eq1}, we have
\begin{align}
&\mathbb{P}(\hat{m} \in \{2,3,...,M\} |m=1) \notag \\
&\le e^{o(k)} (e^{kR}-1) e^{-k(\mathbb{I}(P,W)-\epsilon)}  \sum_{k_1=0}^{\max\{k,n-k\}} \!\! \binom{k}{k_1} \!\!\binom{n-k}{k_1} \!e^{-k d_{k_1 /k}(PW\|W_\star)-(n-k)d_{k_1 /(n-k)}(W_\star\|PW)} \label{E: eq4} \\
& \le e^{o(k)} e^{kR} e^{-k(\mathbb{I}(P,W)-\epsilon)} e^{k f(P,W,\alpha)} \label{E: eq5} \\
& = e^{o(k)} e^{-k \epsilon}  \quad \to 0 \text{ as } k \to \infty, \label{E: eq6}
\end{align}
where:~\eqref{E: eq6} is obtained by substituting $R=\mathbb{I}(P,W)-f(P,W,\alpha)-2\epsilon$; and~\eqref{E: eq5} is obtained by finding the exponent of the sum in~\eqref{E: eq4} as follows
\begin{align}
&\lim_{k \to \infty} \frac{1}{k} \log \sum_{k_1=0}^{\max\{k,n-k\}} \mathsmaller{\binom{k}{k_1} \!\binom{n-k}{k_1}} e^{\mathsmaller{-k d_{k_1/k}(PW\|W_\star)-(n-k)d_{k_1 /(n-k)}(W_\star\|PW)}} \notag \\
&=\lim_{k \to \infty} \!\frac{1}{k} \!\log\! \sum_{k_1=0}^{\max\{k,n-k\}}\! \exp\{\mathsmaller{\!k h(\frac{k_1}{k})\!\!+\!\!(n-k) h(\frac{k_1}{n-k}) \!-\!k d_{\frac{k_1}{k}}(PW\|W_\star)}  \mathsmaller{-(n-k)d_{\frac{k_1}{n-k}}(W_\star\|PW)\}} \label{E: block2-stir} \\
& =\!\lim_{k \to \infty} \!\frac{1}{k} \max_{k_1=0,...,\max\{k,n-k\}} \!\{\mathsmaller{\!k h(\frac{k_1}{k})\!+\!(n-k) h(\frac{k_1}{n-k}) \!-\!k d_{\frac{k_1}{k}}(PW\|W_\star)}  \mathsmaller{-(n-k)d_{\frac{k_1}{n-k}}(W_\star\|PW)\}} \label{E: block2-stir-2} \\
& \le \max_{0 \le \beta \le \max\{1,\frac{1}{\alpha-1}\}} \!\{(\alpha-1)h(\beta)+ h((\alpha-1)\beta) -\!d_{(\alpha-\!1) \beta}(PW \| W_\star)  -(\alpha\!-\!1)d_\beta(W_\star\|PW) \} \label{E: block2-stir-3}  \\
&=f(P,W,\alpha), \label{E: block2-stir-4}
\end{align}
where:~\eqref{E: block2-stir} follows by using Stirling's approximation for the binomial terms;~\eqref{E: block2-stir-2} follows by noticing that the exponent of the summation is equal to the largest exponent of each term in the summation, since the number of terms is polynomial in $k$;~\eqref{E: block2-stir-3} is obtained by letting $\beta:=k_1/(n-k)$ ($0 \le \beta \le \max\{1,\frac{1}{\alpha-1}\}$) and substituting  $n=\alpha k$; and~\eqref{E: block2-stir-4} follows from the definition~\eqref{E: basic-f}.

Now, combining~\eqref{E: union1},~\eqref{E: secterm1}, and~\eqref{E: eq6}, we have $p_e^{avg}\to 0$ as $k \to \infty$, which proves the Theorem.

\section{Proof of Proposition~\ref{prop: overheadf-basic}}
\label{app-proof-prop: overheadf-basic}

\begin{enumerate}[(a)]
\item The term $(\alpha-1)h(\beta)+ h((\alpha-1)\beta)$ in~\eqref{E: basic-f} is maximized at $\beta=1/\alpha$, because it is concave in $\beta$ and its derivative with respect to $\beta$ is zero at $1/\alpha$. Thus, this term is decreasing in $\beta$ in the interval $[1/\alpha,1]$. Also, note that the partial divergence terms in~\eqref{E: basic-f} are increasing with respect to $\beta$ according to Proposition~\ref{prop2}~(\ref{d}). Therefore, the term in the max operator in~\eqref{E: basic-f} is decreasing in $\beta$ in the interval $[1/\alpha,1]$, and the maximum occurs in the interval $[0,1/\alpha]$.
\item The term in the max operator in~\eqref{E: basic-f} is concave in $\beta$, because $h(\beta)$ is concave in $\beta$ and $d_{\beta}(\cdot\|\cdot)$ is convex in $\beta$ according to Proposition~\ref{prop2}~(\ref{e}). Therefore, the term is maximized at a point $\beta^*$ where the derivative with respect to $\beta$ is equal to zero. We then have
    \begin{equation} \label{E: maxbf}
    \log\frac{1-\beta^*}{\beta^*}+\log\frac{1-(\alpha-1)\beta^*}{(\alpha-1)\beta^*} -\log c_1\frac{1-(\alpha-1)\beta^*}{(\alpha-1)\beta^*}-\log c_2\frac{1-\beta^*}{\beta^*}=0,
    \end{equation}
    where:~\eqref{E: derv partial diver} is used to derive~\eqref{E: maxbf}; and $c_1$ and $c_2$ are the corresponding $c^*$'s in~\eqref{E: const c} for the two partial divergence terms in~\eqref{E: basic-f}. Taking derivative of~\eqref{E: basic-f} with respect to $\alpha$ assuming that the maximum occurs at $\beta^*$, we obtain
    \begin{align}
    &\frac{\partial f(P,W,\alpha)}{\partial \alpha} \notag \\
    &= \frac{\partial \beta^*}{\partial \alpha} (\alpha\!-\!1) (\cdot) \!+\!h(\beta^*)\!-\! d_{\beta^*}(W_\star\|PW) \notag \\
    & \quad +\beta^*(\log\frac{1-(\alpha-1)\beta^*}{(\alpha-1)\beta^*}-\log c_1\frac{1-(\alpha-1)\beta^*}{(\alpha-1)\beta^*}) \notag \\
    &=\!h(\beta^*)\!-\! d_{\beta^*}\!(W_\star\|PW)\!+\!\beta^*(\log\! c_2\frac{1-\beta^*}{\beta^*}\!-\!\log\!\frac{1-\beta^*}{\beta^*}\!) \label{E: prop3-1} \\
    &=\!-\log(1-\beta^*)\!+\!\!\left(\!\frac{\partial d_{\beta^*}\!(W_\star\|PW)}{\partial \beta^*}\!-\! d_{\beta^*}\!(W_\star\|PW)\!\!\right) \label{E: prop3-2} \\
    & \ge 0, \label{E: prop3-3}
    \end{align}
    where: $(\cdot)$ in the first line is the left-side of~\eqref{E: maxbf}, which is equal to zero;~\eqref{E: prop3-1} follows from~\eqref{E: maxbf}, and~\eqref{E: prop3-2} follows from~\eqref{E: derv partial diver}; and~\eqref{E: prop3-3} follows from the fact that $-\log(1-\beta^*)$ is always positive for $0 \le \beta^* \le 1$ and $\partial d_{\beta^*}(W_\star\|PW)/\partial \beta^*- d_{\beta^*}(W_\star\|PW)$ is also always positive, because the partial divergence $d_{\beta^*}(W_\star\|PW)$ is convex in $\beta^*$ according to Proposition~\ref{prop2}~(\ref{e}).
\item Substituting $\alpha=1$ in~\eqref{E: basic-f}, all the terms would be zero, because $h(0)=0$ and $d_0(P\|Q)=0$ according to Proposition~\ref{prop2}~(\ref{a}).
\item Consider that the maximum in~\eqref{E: basic-f} occurs at $\beta^*$. According to part (\ref{prop3-a-basic}), $0 \le \beta^* \le 1/\alpha$, and therefore, $\beta^* \to 0$ as $\alpha \to \infty$. Using Proposition~\ref{prop2}~(\ref{c}) and~\eqref{E: maxbf}, we have $\alpha \beta^* \to 1$ as $\alpha \to \infty$. Substitiuting $\alpha=1/\beta^*$ in~\eqref{E: basic-f}, we obtain
    \begin{align}
    &\lim_{\alpha \to \infty}  f(P,W,\alpha) \notag \\
    &=\lim _{\beta^* \to 0} \!\frac{h(\beta^*)}{\beta^*}\!+\! h(1)\!-\! d_{1}(PW \| W_\star)\!-\!\frac{d_{\beta^*}(W_\star\|PW)}{\beta^*} \notag \\
    &=\lim _{\beta^* \to 0} \frac{h(\beta^*)}{\beta^*}- D(PW \| W_\star)-d_0^\prime(P\|Q)    \label{E: prop3-bbb} \\
    &=\lim _{\beta^* \to 0} \frac{h(\beta^*)}{\beta^*}- D(PW \| W_\star)  \label{E: prop3-bbb-2} \\
    & \to \infty,  \label{E: prop3-bbb-3}
    \end{align}
    where:~\eqref{E: prop3-bbb} follows from Proposition~\ref{prop2}~(\ref{b});~\eqref{E: prop3-bbb-2} follows from Proposition~\ref{prop2}~(\ref{c}); and~\eqref{E: prop3-bbb-3} follows from the definition of the binary entropy function and the assumption that  $D(PW\|W_\star)$ is finite.
\end{enumerate}

\section{Proof of Proposition~\ref{prop-auxil1}}
\label{app-proof-prop-auxil1}

We prove the properties for the capacity function $g(a,b)$. The corresponding properties for the function $\phi(a,b)$ easily follows from~\eqref{e-auxiliary-cap2}.
\begin{enumerate}[(a)]
\item Since the cardinality of the input alphabet of this channel is $2^a$, the capacity of this channel is at most $a$ bits per channel use.
\item There are no insertions. Therefore, it is a noiseless channel with input and output alphabets of sizes $2^a$ and capacity $a$ bits per channel use.
\item The input alphabet is $\{0,1\}$, and the output consists of binary sequences with length $b$ and weight $0$ or $1$, because only $0$'s can be inserted in the sequence. Considering all the output sequences with weight $1$ as a super-symbol, the channel becomes binary noiseless with capacity $1$ bits per channel use.
\item The capacity $g(a,b+1)$ cannot decrease if, at each channel use, the decoder knows exactly one of the positions at which an insertion occurs, and the capacity of the channel with this genie-aided encoder and decoder becomes $g(a,b)$. Therefore, $g(a,b+1) \le g(a,b)$.
\item The capacity $g(a+1,b+1)$ cannot decrease if, at each channel use, the encoder and decoder know exactly one of the positions at which an input bit remains unchanged, so that it can be transmitted uncoded and the capacity of the channel with this genie-aided encoder and decoder becomes $1+g(a,b)$. Therefore, $g(a+1,b+1) \le 1+g(a,b)$.
\end{enumerate}

\section*{Acknowledgment}
The authors wish to thank Dr. Aslan Tchamkerten and Shyam Kumar for the useful discussion, which led to a more compact expression and a shorter proof for the result in Lemma~\ref{lemma1}.


\end{document}